\documentclass[11pt]{article}
\usepackage{amsmath} 
\usepackage{amssymb}
\usepackage{amsthm}
\usepackage[ansinew]{inputenc}
\usepackage[T1]{fontenc}
\usepackage{fullpage}
\usepackage{tikz}
\newtheorem{theorem}{Theorem}[section]
\newtheorem{definition}[theorem]{Definition}
\newtheorem{lemma}[theorem]{Lemma}
\newtheorem{corollary}[theorem]{Corollary}
\newtheorem{proposition}[theorem]{Proposition}
\newtheorem{claim}[theorem]{Claim}

\newcommand{\size}{\mathtt{size}}
\newcommand{\vitem}{\vspace*{-0.05in}\item}

\def\nat{{\mathbb N}}
 \def\real{{\mathbb R}}

\begin{document}

\title{Polynomial Time Algorithms for\\
Branching Markov Decision Processes and\\
Probabilistic Min(Max) Polynomial Bellman Equations} 
\author{Kousha Etessami\\U. of Edinburgh\\{\tt kousha@inf.ed.ac.uk}
\and 
Alistair Stewart\\U. of Edinburgh\\{\tt stewart.al@gmail.com} 
\and
Mihalis Yannakakis\\Columbia U.\\{\tt mihalis@cs.columbia.edu}}

\date{}

\maketitle

\begin{abstract}
  We show that one can approximate 
 the least fixed point solution
for a multivariate system of monotone probabilistic
max (min) polynomial equations, 
referred to as maxPPSs (and minPPSs, respectively),
in time  polynomial in both the encoding size of the system of
equations and in 
$\log(1/\epsilon)$, where $\epsilon > 0$ is the desired additive error 
bound of the solution.  (The model of 
computation is the standard Turing machine model.)
We establish this result using a generalization of Newton's method
which applies to maxPPSs and minPPSs,
even though the underlying functions are only piecewise-differentiable.
This generalizes our recent work which provided a P-time 
algorithm for purely probabilistic PPSs.

These equations form the Bellman optimality equations 
for several important classes of {\em infinite-state} Markov Decision
Processes (MDPs).  Thus,
as a corollary, we obtain the
first polynomial time algorithms for computing
to within arbitrary desired precision the {\em optimal value} vector
for several classes of infinite-state MDPs
which arise as extensions of 
classic, and heavily studied, purely stochastic processes.
These include both the problem of maximizing and
mininizing the {\em termination (extinction) probability} of  multi-type branching MDPs,
stochastic context-free MDPs, and 1-exit Recursive MDPs.

Furthermore, we also show that we can compute in P-time 
an $\epsilon$-optimal  
policy for both maximizing and minimizing 
branching, context-free, and 1-exit-Recursive MDPs,
for any given desired $\epsilon > 0$.
This is despite the fact that actually computing optimal strategies is
{\em Sqrt-Sum}-hard and {\em PosSLP}-hard in this setting.

We also derive, as an easy consequence of these results, an FNP 
upper bound on the complexity
of computing the value (within arbitrary desired precision) of 
branching simple stochastic games (BSSGs) and related 
infinite-state turn-based stochastic
game models. 
\end{abstract}

\section{Introduction}
Markov Decision Processes (MDPs) are a fundamental model for stochastic dynamic
optimization and optimal control, with applications in many fields.  
They extend purely stochastic processes (Markov chains) with a
controller (an agent) who can partially affect the evolution of
the process, and seeks to optimize some objective.
For many important classes of MDPs, the task of
computing the {\em optimal value} of the objective, 
starting at any state of the MDP,
can be rephrased as the problem of solving
the associated {\em Bellman optimality equations}
for that MDP model.   In particular, for finite-state
MDPs where, e.g., the objective is to maximize (or minimize) the probability
of eventually reaching some target state, the associated Bellman
equations are {\em max-(min-)linear} equations, and we know how
to solve such equations in P-time using linear programming 
(see, e.g., \cite{Puterman94}).  The same holds for a number of other
classes of finite-state MDPs.

In many important settings however, the state space of the processes
of interest, both for purely stochastic processes, as well as for
controlled ones (MDPs), is not finite, even though the processes can
be specified in a finite way.  For example, consider {\em
  multi-type branching processes} (BPs) \cite{KolSev47,Harris63}, a
classic probabilistic model with applications in many areas
(biology, physics, etc.).  A BP models the stochastic evolution of a
population of entities of distinct types.  In each generation, every
entity of each type $T$ produces a set of entities of various types in
the next generation according to a given probability distribution on
offsprings for the type $T$.  In a {\em Branching Markov Decision Process}
(BMDP) \cite{pliska76, rotwhit82}, there is a controller who can  take
actions that affect the probability distribution for the sets of
offsprings for each entity of each type.  
For both BPs and BMDPs, 
the state space consists of all possible populations, given by the number
of entities of the various types, so there are an infinite number of states.
From the computational point of view,
the usefulness of such infinite-state models 
hinges on whether their analysis remains tractable.

In recent years there has been a body of research
aimed at studying the computational complexity of key
analysis problems associated
with MDP extensions (and, more general stochastic game extensions)
of important classes of finitely-presented but
{\em countably infinite-state} stochastic processes, including
controlled extensions of
classic multi-type branching processes (i.e., BMDPs),
and {\em stochastic context-free grammars}, 
and discrete-time {\em quasi-birth-death processes}.
In \cite{rmdp} a model called {\em recursive Markov decision
processes} (RMDP) was studied 
that is in a precise sense 
more general than all of these, and
forms the MDP extension of {\em recursive Markov chains} \cite{rmc}
(and equivalently, {\em probabilistic pushdown systems} \cite{EKM}),
or it can be  viewed alternatively as the extension
of finite-state MDPs with recursion.

A central analysis problem for all of these models,
which  forms the key to a number of other analyses, is 
the problem of computing their 
{\em optimal termination (extinction) probability}.
For example, in the setting of multi-type Branching MDPs (BMDPs), 
these key quantities are 
the maximum (minimum) probabilities, over all control strategies (or policies), 
that starting from a single entity of a given type, the process will eventually
reach extinction (i.e., the state where no entities have survived).
From these quantities, one can compute the  optimum probability for
any initial population, as well as other quantities of interest.

One can indeed form Bellman optimality equations
for the  optimal extinction probabilities of BMDPs,
and for a number of related important infinite-state MDP models.  
However, it turns out that these optimality equations are no longer 
max/min {\em linear} but rather are max/min {\em polynomial} equations (\cite{rmdp}).
Specifically, the Bellman equations for BMDPs
with the objective of maximizing (or minimizing) extinction probability
are multivariate systems of monotone probabilistic max (or  min) 
polynomial equations, which we call {\bf max/minPPSs}, of the form
$x_i=P_i(x_1,\ldots,x_n)$, $i=1,\ldots,n$, 
where each $P_i(x) \equiv \max_j q_{i,j}(x)$ 
(respectively $P_i(x) \equiv \min_j q_{i,j}(x)$) is the max (min) over
a finite number of probabilistic polynomials, $q_{i,j}(x)$.
A {\em probabilistic polynomial}, $q(x)$, is a multi-variate polynomial
where the monomial coefficients and constant term of $q(x)$ 
are all non-negative and sum to $\leq 1$.
We write these equations in vector form as $x = P(x)$.
Then $P(x)$ defines  a mapping $P:[0,1]^n \rightarrow [0,1]^n$
that is monotone, and thus (by Tarski's theorem) 
has a {\em least fixed point} in
$[0,1]^n$.
The equations $x=P(x)$, can have more than one solution, but it turns
out that the 
optimal value vector for the corresponding BMDP is precisely the
least fixed point (LFP) solution vector $q^* \in [0,1]^n$,
i.e., the (coordinate-wise) least non-negative solution 
(\cite{rmdp}).

Already for pure stochastic multi-type branching processes (BPs),
the extinction probabilities may be irrational values.
The problem of {\em deciding}
whether the extinction probability of a BP is $\geq p$, 
for a given probability $p$ is in PSPACE (\cite{rmc}), and
likewise, deciding whether the optimal extinction probability
of a BMDP is $\geq p$ is in PSPACE (\cite{rmdp}).
These PSPACE upper bounds appeal to decision procedures
for the existential theory of reals for solving the associated 
(max/min)PPS equations.
However, already for BPs,
it was shown in \cite{rmc} that this quantitative {\em decision} problem is already
at least as hard as the {\em square-root sum} problem,
as well as a (much) harder and more fundamental 
problem called {\em PosSLP}, which captures the power of
unit-cost exact rational arithmetic.  It is 
a long-standing open problem whether
either of these decision problems is in NP,
or even in the polynomial time hierarchy (see \cite{ABKM06,rmc} for more information on these problems).   
Thus, such {\em quantitative  decision problems} are unlikely
to have P-time algorithms, even in the purely stochastic setting,
so we can certainly not expect to find P-time algorithms for the extension
of these models to the MDP setting.
On the other hand, it was shown in \cite{rmc} and \cite{rmdp},
that for both BPs and BMDPs the {\em qualitative} decision problem
of deciding whether the optimal extinction 
probability $q^*_i = 0$ or whether $q^*_i = 1$, can be solved
in polynomial time.

Despite decades of theoretical and practical
work on computational problems like extinction relating to 
multi-type branching processes, and 
equivalent termination problems related to stochastic context-free grammars,
until recently it was not even known whether one could
obtain {\em any} non-trivial
{\em approximation} of the extinction probability of a purely stochastic
multi-type branching processes (BP) in P-time.
The extinction probabilities of pure BPs are the LFP of a system of 
probabilistic polynomial equations (PPS), without max or min.
In recent work \cite{ESY12}, we provided the first polynomial
time algorithm for computing (i.e., approximating) to within
any desired additive error $\epsilon > 0$ the LFP of a given PPS,
and hence the extinction probability vector $q^*$
for a given pure stochastic BP, 
in time polynomial in both the encoding
size of the PPS (or the BP) and in $\log(1/\epsilon)$.
The algorithm works in the standard Turing model of computation.
Our algorithm was based on an approach using 
Newton's method that was first introduced and studied in \cite{rmc}.
In \cite{rmc} the approach was studied for more
general systems of {\em monotone} polynomial equations (MPSs), and
it was subsequently further studied in \cite{lfppoly}. 

Note that unlike PPSs and MPSs, the min/maxPPSs that define
the Bellman equations for BMDPs
are no longer differentiable 
functions (they are only piecewise differentiable).
Thus, a priori, it is not even clear how one could apply
a Newton-type method toward solving them.

In this paper we extend the results of \cite{ESY12}, and 
provide the first polynomial time algorithms for approximating
the LFP of both maxPPSs and minPPSs, and thus the first polynomial
time algorithm for computing (to within any desired additive error) 
the optimal value vector for
BMDPs with the objective of maximizing or minimizing their extinction probability.

Our approach is based on a {\em generalized Newton's method} (GNM), that 
extends Newton's method in a natural way to the setting of max/minPPSs,
where each iteration requires the computation of the 
least (greatest) solution of a 
max- (min-) linear system of equations, both of which we show can be solved using linear
programming.   Our approach also makes crucial use
of the P-time algorithms in \cite{rmdp} for {\em qualitative}
analysis of max/min BMDPs, which allow us to remove 
variables $x_i$  where the LFP is $q^*_i = 1$ or where $q^*_i =0$.
The algorithms themselves have the nice feature that they
are relatively simple, although the analysis of their correctness and
time complexity is rather involved.

We furthermore show that we can compute $\epsilon$-optimal
(pure) strategies (policies) for both
maxPPSs and minPPSs, for {\em any} given desired $\epsilon > 0$, in 
time polynomial in both the encoding size of the max/minPPS and
in $\log(1/\epsilon)$.
This result is at first glance rather surprising, because
there are only a bounded number of distinct pure policies
for a max/minPPS, and computing an optimal policy is PosSLP-hard.
The proof of this result involves an intricate analysis of
bounds on the norms of certain matrices associated with (max/min)PPSs.

Finally, we consider
{\em Branching simple stochastic games} (BSSGs),
which are two-player turn-based stochastic games,
where one player wants to maximize, and the other wants to minimize,
the extinction probability (see \cite{rmdp}).   
The {\em value} of these games (which are determined) 
is characterized by the LFP solution of associated min-maxPPSs
which combine both min and max operators (see \cite{rmdp}).
We observe that 
our results easily imply a FNP upper bound for $\epsilon$-approximating the {\em value}
of BSSGs and computing $\epsilon$-optimal strategies for them.

\medskip

\noindent {\bf Related work:} 
We have already mentioned some of the important relevant results.
BMDPs and related processes have been studied previously in both
the operations research (e.g. \cite{pliska76,rotwhit82,denrot05a}) 
and computer science literature (e.g. \cite{rmdp,EGKS08,BBFK08}), but no 
efficient algorithms
were known for the (approximate) computation of the relevant optimal
probabilities and policies; the best known upper bound was PSPACE \cite{rmdp}.

In \cite{rmdp} we introduced Recursive Markov Decision Processes (RMDPs), 
a recursive extension of MDPs. 
We showed that for general RMDPs, the problem of computing the optimal termination probabilities,
even within any nontrivial approximation, is undecidable. However, 
we showed for the
important class of 1-exit RMDPs (1-RMDP), the optimal probabilities can be expressed 
by min (or max) PPSs, and in fact the problems of computing (approximately) the
LFP of a min/maxPPS and the termination probabilities of  a max/min 1-RMDP, 
or BMDP, are all polynomially equivalent. We furthermore showed 
in \cite{rmdp} that there 
are always pure, memoryless
optimal policies for both maximizing and minimizing 1-RMDPs (and for
the more general turn-based stochastic games). 

In \cite{EWY08}, 1-RMDPs with a different objective were studied, 
namely optimizing the total expected reward in a setting
with positive rewards. In that setting, things are much simpler:
the Bellman equations turn out to be max/min-linear, 
the optimal values are rational, and they
can be computed {\em exactly} in P-time using linear programming.

A work that is more closely related to this paper is \cite{EGKS08}
by
Esparza, Gawlitza, Kiefer, and Seidl. 
They studied more general monotone
min-maxMPSs, i.e., systems of monotone polynomial equations that
include both min and max operators, and they presented two 
different iterative 
analogs of Newton's methods
for approximating the LFP of a min-maxMPS, $x=P(x)$.
Their methods are related to ours, but differ in key respects.
Both of their methods use certain piece-wise linear functions 
to approximate the min-maxMPS 
in each iteration, which is also what one does to
solve  each
iteration of our generalized Newton's method.
However, the precise nature of their piece-wise linearizations,
as well as how they solve them, differ 
in important ways from ours, even when they are applied in the
specific context of maxPPSs or minPPSs.  
They show, working in the unit-cost {\em exact} arithmetic model,
that using their methods one can compute
$j$ ``valid bits'' of the LFP (i.e., compute
the LFP within relative error
at most $2^{-j}$)  
in $k_P + c_P\cdot j$ 
iterations, where $k_P$ and $c_P$ are
terms that depend in {\em some} way on the input system, $x=P(x)$.
However, they give no 
constructive upper bounds on $k_P$, and their upper bounds 
on $c_P$ are exponential in the number
$n$ of variables of $x=P(x)$.
Note that MPSs are more difficult:
even without the min and max operators, we know that it is PosSLP-hard
to approximate their LFP within any nontrivial constant additive error $c<1/2$,
even for pure MPSs that arise from Recursive Markov Chains \cite{rmc}.

Another subclass of RMDPs,  called
{\em one-counter MDPs} (a controlled extension of one-counter Markov chains
and Quasi-Birth-Death processes \cite{EWY10}) has been studied,
and the approximation of their  optimal termination probabilities 
was recently shown to be computable,
but only in {\em exponential time}  (\cite{BBEK11}).
This subclass is incomparable with 1-RMDPs and BMDPs, and 
does not have min/maxPPSs as Bellman equations.

\section{Definitions and Background}
For an $n$-vector of variables $x = (x_1,\ldots,x_n)$, and 
a vector $v \in \nat^n$,  we use the shorthand notation $x^v$ to denote
the monomial $x_1^{v_1}\ldots x^{v_n}_{n}$.
Let $\langle \alpha_r \in \nat^n \mid r \in R \rangle$ be a multi-set
of $n$-vectors of natural numbers, indexed by the set $R$.
Consider a multi-variate polynomial $P_i(x) = \sum_{r \in R} p_r x^{\alpha_r}$,
for some rational-valued coefficients $p_r$, $r \in R$.
We shall call $P_i(x)$ a {\bf\em monotone polynomial} if  
$p_r \geq 0$ for all $r \in R$.
If in addition, we also have   $\sum_{r \in R} p_r \leq 1$,
then we shall call $P_i(x)$ a {\bf\em probabilistic polynomial}.

\begin{definition} A {\bf probabilistic} (respectively, {\bf monotone}) {\bf polynomial system of equations}, 
$x = P(x)$, which we shall call
a  {\bf PPS} (respectively, a {\bf MPS}),  
is a system of $n$ equations, $x_i = P_i(x)$,  in $n$ variables $x = (x_1,x_2,...,x_n)$,
  where for all $i \in \{ 1,2,...n\}$,  $P_i(x)$ is a probabilistic (respectively, monotone) polynomial.

A {\bf maximum-minimum probabilistic polynomial system of equations}, $x=P(x)$,
called a  {\bf max-minPPS} 
is a system of $n$ equations in $n$ variables $x= (x_1,x_2,\ldots,x_n)$,  
where for all $i \in \{1,2,\ldots,n\}$, either: 
\begin{itemize}
\item {\tt Max-polynomial}: $P_i(x) = \max  \{q_{i,j}(x): j \in \{1,...,m_i\}\}$,   Or: 

\item {\tt Min-polynomial}: $P_i(x) = \min  \{q_{i,j}(x): j \in \{1,...,m_i\}\}  \ $
\end{itemize}

\noindent where each $q_{i,j}(x)$ is a probabilistic polynomial, for every 
$j \in \{1,\ldots,m_i\}$.

We shall call such a system a {\bf maxPPS} (respectively, a {\bf minPPS})
if for every $i \in \{1,\ldots,n\}$, $P_i(x)$ is a {\tt Max-polynomial} (respectively, a {\tt Min-polynomial}).

Note that we can view a PPS in n variables 
as a maxPPS, or as a minPPS, where $m_i =1$ for every $i \in \{1,\ldots,n\}$.
\end{definition}

For computational purposes we assume that all the coefficients are rational.
We assume that the polynomials in a system are given in sparse form,
i.e., by listing only the nonzero terms, with the coefficient and the nonzero exponents 
of each term given in binary. 
We let $|P|$ denote the total bit encoding length of a system $x=P(x)$
under this representation.

We use  {\bf max/minPPS} to refer to a system of equations, $x=P(x)$, 
that is either a maxPPS or a minPPS. 
While \cite{rmdp} 
also considered systems of equations containing both max and min equations
(which we refer to as {\bf max-minPPSs}), our primary focus will 
be on systems that contain just one or the other.
(But we shall also obtain results about max-minPPSs as a corollary.) 

As was shown in \cite{rmdp}, any max-minPPS, $x=P(x)$, 
has a {\bf least fixed point} ({\bf LFP}) 
solution, $q^* \in [0,1]^n$, i.e.,  $q^* = P(q^*)$ and if $q = P(q)$ for some $q \in [0,1]^n$ then $q^* \leq q$
(coordinate-wise inequality).  As observed in \cite{rmc, rmdp},  $q^*$ may in general
contain irrational values, even in the case of PPSs.
The central results of this paper  yield P-time algorithms for computing
$q^*$ to within arbitrary precision, both in the case of maxPPSs and minPPSs.
As we shall explain,
our P-time upper bounds for computing (to within any desired accuracy) the 
least fixed point of maxPPSs and minPPSs will also yield, as corollaries,
FNP upper bounds for computing approximately the LFP of max-minPPSs.

\begin{definition} We define a  {\bf policy} 
for a max/minPPS, $x=P(x)$, to be a function 
$\sigma: \{1,...n\} \rightarrow \mathbb{N}$ 
such that $1 \leq \sigma(i) \leq m_i$.
\end{definition}

Intuitively, for each variable, $x_i$, a policy selects 
one of the probabilistic polynomials, $q_{i,\sigma(i)}(x)$,
that appear on the RHS of the equation $x_i = P_i(x)$, and
which $P_i(x)$ is the maximum/minimum over.

\begin{definition} Given a max/minPPS $x=P(x)$ over $n$ variables, and a policy $\sigma$ for $x=P(x)$, 
we define the PPS $x=P_\sigma(x)$ by:
$$(P_\sigma)_i(x) = q_{i,\sigma(i)}$$
for all $i \in \{1,\ldots,n\}$.
\end{definition}

Obviously, since a PPS is a special case of a max/minPPS, 
every PPS also has a unique LFP solution (this was established earlier in \cite{rmc}). 
Given a max/minPPS, $x =P(x)$, and a policy, $\sigma$,
we use $q^*_\sigma$ to denote the LFP solution vector 
for the PPS $x=P_\sigma(x)$.

\begin{definition} For a maxPPS, $x=P(x)$, a policy $\sigma^*$ 
is called {\bf optimal} if for all other policies $\sigma$, 
$q^*_{\sigma^*} \geq q^*_\sigma$.
For a minPPS $x=P(x)$ a policy $\sigma^*$ is optimal if for all other policies 
$\sigma$, $q^*_{\sigma^*} \leq q^*_\sigma$.
A policy $\sigma$ is {\bf $\epsilon$-optimal} for $\epsilon >0$ if
$||q^*_{\sigma} -q^*||_{\infty} \leq \epsilon$.
\end{definition}

\noindent A non-trivial fact is that optimal policies always exist, and 
furthermore that they actually attain the
LFP $q^*$ of the max/minPPS:

\begin{theorem}[\cite{rmdp}, Theorem 2] \label{optexist} 
For any max/minPPS, $x=P(x)$, 
there always exists an optimal policy $\sigma^*$,
and furthermore $q^* = q^*_{\sigma^*}$.\footnote{Theorem 2 
of \cite{rmdp} is stated in the more general context of 
1-exit Recursive Simple
Stochastic Games and shows that also for max-minPPSs, both
the max player and the min player have optimal policies that 
attain the LFP $q^*$.} 
\end{theorem}

Probabilistic polynomial systems can be used to capture central probabilities 
of interest for several basic stochastic models, 
including Multi-type Branching Processes (BP), 
Stochastic Context-Free Grammars (SCFG) and the class of 1-exit Recursive Markov Chains (1-RMC) \cite{rmc}. 
Max- and minPPSs can be similarly used to capture the central optimum probabilities of corresponding stochastic optimization models:
(Multi-type) Branching Markov Decision processes (BMDP), Context-Free MDPs (CF-MDP), 
and 1-exit Recursive Markov Decision Processes (1-RMDP) \cite{rmdp}.
We now define BMDPs and 1-RMDPs.

\medskip

A {\bf Branching Markov Decision Process} (BMDP) consists of a finite set 
$V=\{T_1, \ldots, T_n\}$ of types, a finite set $A_i$ of actions for each type,
and a finite set $R(T_i,a)$ of probabilistic rules for each type $T_i$ 
and action $a_i \in A_i$. Each rule $r \in  R(T_i,a)$ has the form
$T_i \stackrel{p_r}{\rightarrow} \alpha_r$, where $\alpha_r$ is a finite multi-set 
whose elements are in $V$, $p_r \in (0,1]$ is the probability of the rule,
and the sum of the probabilities of all the rules in $R(T_i,a)$ is equal to 1:
$ \sum_{r \in R(T_i,a) } p_r = 1 $. 

Intuitively, a BMDP describes the stochastic evolution of entities of given types
in the presence of a controller that can influence the evolution. Starting from an
initial population (i.e. set of entities of given types) $X_0$ at time (generation) 0,
a sequence of populations $X_1, X_2, \ldots$ is generated, where $X_k$ is obtained from
$X_{k-1}$ as follows. 
First the controller selects for each entity of $X_{k-1}$ an available action for the type of the entity; 
then a rule is chosen independently and simultaneously for every entity of $X_{k-1}$ probabilistically 
according to the probabilities of the rules for the type of the entity and the selected action, 
and the entity is replaced by a new set of entities with the types specified by the right-hand side
of the rule. The process is repeated as long as the current population $X_k$ is nonempty, 
and terminates if and when $X_k$ becomes empty. 
The objective of the controller is either to minimize the probability of termination (i.e., extinction of
the population), in which case the process is a minBMDP, or to maximize the termination probability, 
in which case it is a maxBMDP. 
At each stage, $k$, the controller is allowed in principle to select the actions for the entities of $X_k$ 
based on the whole past history, may use randomization (a mixed strategy) and may make different choices 
for entities of the same type. 
However, it turns out that these flexibilities do not increase the controller's power, 
and there is always an optimal pure, 
memoryless strategy that always chooses the same action 
for all entities of the same type (\cite{rmdp}). 

For each type $T_i$ of a minBMDP (respectively, maxBMDP), let $q^*_i$ 
be the minimum (resp. maximum) probability of termination 
if the initial population consists of a single entity of type $T_i$. 
From the given minBMDP (maxBMDP) we can construct
a minPPS (resp. maxPPS) $x=P(x)$ whose LFP is precisely the vector $q^*$ of optimal 
termination (extinction) probabilities (see Theorem 20 in the full version of \cite{rmdp}): 
The min/max polynomial $P_i(x)$ for each type $T_i$ contains one polynomial
$q_{i,j}(x)$ for each action $j \in A_i$, with 
$q_{i,j}(x) = \sum_{r \in R(T_i,j)} p_r x^{\alpha_r}$.

\medskip

A {\bf 1-exit Recursive Markov Decision Process} (1-RMDP) consists of
a finite set of components $A_1, \ldots , A_k$, where each component $A_i$
is essentially a finite-state MDP augmented with the ability to
make recursive calls to itself and other components.
Formally, each component $A_i$ has a finite set $N_i$ of nodes, which are partitioned
into probabilistic nodes and controlled nodes, and a finite set $B_i$ of "boxes"
(or supernodes), where each box is mapped to some component.
One node $en_i$ is specified as the entry of the component $A_i$ and one node $ex_i$
as the exit of $A_i$.\footnote{The restriction to having only one entry node is not important;
any multi-entry RMDP can be efficiently transformed to an 1-entry RMDP.
The restriction to 1-exit is very important: 
multi-exit RMDPs lead to undecidable termination problems, 
even for any non-trivial approximation of the optimal values \cite{rmdp}.}
The exit node has no outgoing edges.
All other nodes and the boxes have outgoing edges; the edges
out of the probabilistic nodes and boxes are labelled with probabilities,
where the sum of the probabilities out of the same node or box is equal to 1.

Execution of a 1-RMDP starts at some node,
for example, the entry $en_1$ of component $A_1$.
When the execution is at a probabilistic node $v$, then an edge out of $v$
is chosen randomly according to the probabilities of the edges out of $v$.
At a controlled node $v$, an edge out of $v$ is chosen by a controller who seeks
to optimize his objective. When the execution reaches a box $b$ of $A_i$ mapped to
some component $A_j$, then the current component is suspended and
a recursive call to $A_j$ is initiated at its entry node $en_j$; if the call to $A_j$
terminates, i.e. reaches eventually its exit node $ex_j$, 
then the execution of component $A_i$ resumes from box $b$ following an edge out of $b$ 
chosen according to the probability distribution of the outgoing edges of $b$.
Note that a call to a component can make further recursive calls,
thus, at any point there is in general a stack of suspended recursive calls,
and there can be an arbitrary number of such suspended calls;
thus, a 1-RMDP induces generally an infinite-state MDP.
The process terminates when the execution reaches the exit of the component
of the initial node and there are no suspended recursive calls.

There are two types of 1-RMDPs with a termination objective:
In a min 1-RMDP (resp. max 1-RMDP) the objective of the controller is
to minimize (resp. maximize) the probability of termination.
In principle, a controller can use the complete past history of the process
and also use randomization (i.e. a mixed strategy) to select at each point
when the execution reaches a controlled node which edge to select out of the node.
As shown in \cite{rmdp} however, there is always an optimal strategy
that is pure, stackless and memoryless, i.e., selects deterministically one edge
out of each controlled node, the same one every time, 
independent of the stack and of the past history
(including the starting node).
From a given min or max 1-RMDP we can construct efficiently 
a minPPS or maxPPS, whose LFP yields the
minimum or maximum termination probabilities for all the
different possible starting vertices of the 1-RMDP \cite{rmdp}.
Conversely, from any given min/max PPS, we can efficiently construct a 1-RMDP
whose optimal termination probabilities yield the LFP of the min/max PPS. 
The system of equations for a 1-RMDP has a particularly simple form.
All max/minPPS can be put in that form.

It is convenient to put max/minPPS in the following simple form.

\begin{definition} A maxPPS in {\bf simple normal form (SNF)}, $x= P(x)$, 
is a system of $n$ equations in $n$ variables $x_1,x_2,...x_n$  where each $P_i(x)$ for $i = 1,2,...n$ 
is in one of three forms:
\begin{itemize}
\item {\tt Form L}: $P(x)_i = a_{i,0} + \sum_{j=1}^n a_{i,j} x_j$,
where  $a_{i,j} \geq 0$ for all $j$, and  such that $\sum_{j=0}^n a_{i,j} \leq 1$
\item {\tt Form Q}:  $P(x)_i = x_j x_k$ for some $j,k$
\item {\tt Form M}: $P(x)_i =\max  \{ x_j ,  x_k \}$ for some $j,k$
\end{itemize}
We define {\bf SNF form} for minPPSs analogously: only the definition of ``{\tt Form M}'' changes, 
replacing $\max$ with $\min$.
\end{definition}

In the setting of a max/minPPS in SNF form, for simplicity in notation, when 
we talk about a policy, if $P_i(x)$ has form $M$, say $P_i(x) \equiv \max \{ x_j, x_k\}$,
then when it is clear from the context we will use $\sigma(i) = k$ to mean that the 
policy $\sigma$ chooses $x_k$ among the two choices $x_j$ and $x_k$ 
available in $P_i(x) \equiv \max \{x_j,x_k\}$.

\begin{proposition}[cf. Proposition 7.3 \cite{rmc}]
\label{prop:snf-form}
Every max/minPPS, $x = P(x)$, can be transformed in P-time
to an ``{\em equivalent}''  max/minPPS ,  $y=Q(y)$  in SNF form,
such that  $|Q| \in O( |P|  )$.
More precisely, the variables $x$ are 
a subset of the variables $y$, the LFP of  $x = P(x)$
is the projection of the LFP of  $y=Q(y)$,
and an optimal policy (respectively, $\epsilon$-optimal
policy) for  $x = P(x)$ can be obtained in P-time from an
optimal (resp., $\epsilon$-optimal) policy of  $y=Q(y)$.
\end{proposition}
\begin{proof}
We can easily convert, in P-time, any max/minPPS  into SNF form, 
using the following procedure.
\begin{itemize}
\item For each equation $x_i = P_i(x) = \text{max } \{p_1(x), \ldots ,p_m(x)\}$, 
for each $p_j(x)$ on the right-hand-side that is not a variable, add a new variable 
$x_k$, replace $p_j(x)$ with $x_k$ in $P_i(x)$, 
and add the  new equation $x_k = p_j(x)$. 
Do similarly if $P_i(x) = \min \{p_1(x), \ldots, p_m(x) \}$.

\item If $P_i(x) = \text{max } \{x_{j_1},...,x_{j_m}\}$ with $m >2$, then add $m-2$
new variables $x_{i_1}, \ldots, x_{i_{m-2}}$, 
set $P_i(x) =  \text{max } \{x_{j_1},x_{i_1}\}$,
and add the equations $x_{i_1} = \text{max } \{x_{j_2},x_{i_2}\}$, 
$x_{i_2} = \text{max } \{x_{j_3},x_{i_3}\}$, $\ldots$, 
$x_{i_{m-2}} = \text{max } \{x_{j_{m-1}},x_{j_m}\}$. 
Do similarly if $P_i(x) = \min \{x_{j_1},...,x_{j_m}\}$ with $m >2$.

\item For each equation $x_i =P_i(x) = \sum_{j=1}^m p_j x^{\alpha_j}$,
where $P_i(x)$ is a probabilistic polynomial that is not just a 
constant or a single monomial, 
replace every monomial $x^{\alpha_j}$ on the right-hand-side
that is not a single variable by a new variable $x_{i_j}$ 
and add the equation  $x_{i_j}=x^{\alpha_j}$.

\item For each variable $x_i$ that occurs in some polynomial with
exponent higher than 1, introduce new variables 
$x_{i_1}, \ldots, x_{i_k}$ where $k$ is the logarithm of
the highest exponent of $x_i$ that occurs in $P(x)$,
and add equations $x_{i_1} = x_i^2$,  $x_{i_2} = x_{i_1}^2$, $\ldots$,
$x_{i_k} = x_{i_{k-1}}^2$.
For every occurrence of a higher power $x_i^l$, $l>1$, of $x_i$ in $P(x)$,
if the binary representation of the exponent $l$ is $a_k \dots a_2 a_1 a_0$,
then we replace $x_i^l$ by the product of the variables $x_{i_j}$  such that the
corresponding bit $a_j$ is 1, and $x_i$ if $a_0=1$.
After we perform this replacement for all the higher powers of
all the variables, every polynomial of total degree >2 is just a product of variables.

\item If a polynomial $P_i(x) = x_{j_1} \cdots x_{j_m}$ in the current system
is the product of $m>2$ variables, then add $m-2$
new variables $x_{i_1}, \ldots, x_{i_{m-2}}$, 
set $P_i(x) = x_{j_1}x_{i_1}$,
and add the equations $x_{i_1} = x_{j_2} x_{i_2}$, 
$x_{i_2} = x_{j_3} x_{i_3}$, $\ldots$, 
$x_{i_{m-2}} = x_{j_{m-1}} x_{j_m}$. 
\end{itemize}
Now all equations are of the form L, Q, or M.

The above procedure allows us to 
convert any max/minPPS into one in SNF form by introducing $O(|P|)$ new variables and 
blowing up the size of $P$ by a constant factor $O(1)$. 
Furthermore, there is an obvious (and easy to compute) bijection between  
policies for the resulting SNF form max/minPPS and the original max/minPPS.
\end{proof}

Thus from now on, and for the
rest of this paper {\em we assume, without loss of generality, that all max/minPPSs are in 
SNF normal form. }

We now summarize some of the main previous results on PPSs and max/minPPSs.

\begin{proposition}[\cite{rmdp}]
\label{prob1-ptime-scfg-prop}
There is a P-time algorithm that, given a minPPS or maxPPS, $x=P(x)$, over $n$ variables,
with LFP $q^* \in \real^n_{\geq 0}$,
determines for every $i = 1,\ldots,n$ whether $q^*_i = 0$ or $q^*_i = 1$
or $0< q^*_i <1$.
\end{proposition}

Thus, given a max/minPPS we can find in P-time all the variables $x_i$ such that
 $q^*_i = 0$ or $q^*_i = 1$, remove them and their corresponding equations
$x_i = P_i(x)$, and substitute  their values on the RHS of the remaining equations.
This yields a new max/minPPS, $x' = P'(x')$, where its LFP solution, 
$q'^*$, is ${\textbf 0} < q'^* < {\textbf 1}$, which corresponds to 
the remaining coordinates of $q^*$. 
Thus, it suffices to focus our attention to systems whose LFP is strictly between 0 and 1.

The decision problem of determining whether a coordinate $q^*_i$ of the LFP
is $\geq 1/2$ (or whether $q^*_i \geq r$ for any other given bound $r \in (0,1)$) is at least as hard as
the Square-Root-Sum and the PosSLP problems even for PPS (without the min
and max operator) \cite{rmc} and hence it is highly unlikely that it can be solved in P.

The problem of approximating efficiently the LFP of a PPS was solved recently in \cite{ESY12},
by using Newton's method  after elimination of the variables with value 0 and 1.

\begin{definition} For a PPS $x=P(x)$ we use $P'(x)$ to denote the Jacobian matrix of partial derivatives of $P(x)$,
i.e., $P'(x)_{i,j} := \frac{\partial P_i(x)}{\partial x_j}$. 
For a point $x \in \mathbb{R}^n$, if $(I-P'(x))$ is 
non-singular, then we define one Newton iteration at $x$ via the operator:
$$\mathcal{N}(x) = x + (I-P'(x))^{-1}(P(x) -x)$$
Given a max/minPPS, x=P(x), and a policy $\sigma$,  
we use $\mathcal{N}_\sigma(x) $ to denote the Newton operator
of the PPS $x=P_{\sigma}(x)$; 
i.e., if $(I-P'_\sigma(x))$ is 
non-singular at a point $x \in \mathbb{R}^n$,
then $\mathcal{N}_\sigma(x) = x + (I-P'_\sigma(x))^{-1}(P_\sigma(x) -x)$.
\end{definition}

\begin{theorem}[Theorem 3.2 and Corollary 4.5 of \cite{ESY12}] Let $x=P(x)$ be a PPS 
with rational coefficients in SNF form
which has least fixed point $0 < q^* < 1$.
If we conduct iterations
of  Newton's method as follows: $x^{(0)} := 0$, and for $k \geq 0$: 
$x^{(k + 1)} := \mathcal{N}(x^{(k)})$, 
then the Newton operator $\mathcal{N}(x^{(k)})$ is defined for all $k \geq 0$, and for any $j>0$:
$$\|q^* - x^{(j + 4|P|)}\|_\infty \leq 2^{-j} $$
where $|P|$ is the total bit encoding length of the system  $x=P(x)$.

Furthermore, there is an algorithm (based on 
suitable rounding of Newton iterations)
which, given a PPS, $x=P(x)$, and given a positive integer $j$,
computes a rational vector $v \in [0,1]^n$, such
that $||q^*-v ||_\infty \leq 2^{-j}$, 
and which runs in time polynomial in $|P|$ and $j$
in the standard Turing model of computation.
\end{theorem}

The proof of the theorem involves a number of technical lemmas on PPS and Newton's method, 
several of which we will also need in this paper, some of them in strengthened form.

\begin{lemma}{(c.f., Lemma 3.6 of \cite{ESY12})}
\label{newtbounds} 
Given a PPS, $x=P(x)$, with LFP $q^* > 0$, 
if $0 \leq y \leq q^*$, and if $(I-P'(y))^{-1}$ exists and is non-negative
(in which case clearly $\mathcal{N}(y)$ 
is defined),   
then $\mathcal{N}(y) \leq q^*$ holds.\footnote{Note that the Lemma does
not claim that $\mathcal{N}(y) \geq 0$ holds.  Indeed, it may not.}
\end{lemma}
\begin{proof}

In Lemma 3.4 of \cite{ESY12} it was established that when 
$(I-P'(y))$ is non-singular,
i.e., $(I-P'(y))^{-1}$ is defined, and thus $\mathcal{N}(y)$ 
is defined, then 
\begin{equation}\label{eqn:newt-diff}
q^* - {\mathcal N}(y) =  (I - P'(y))^{-1} \frac{P'(q^*)-P'(y)}{2} (q^*-y)
\end{equation}
Now, since all polynomials in $P(x)$ have non-negative coefficients,
it follows that the Jacobian $P'(x)$ is monotone in $x$, and thus 
since $y \leq q^*$,
we have that $P'(q^*) \geq P'(y)$.
Thus $(P'(q^*) - P'(y)) \geq 0$, and by assumption $(q^*-y) \geq 0$.
Thus, by the assumption that $(I - P'(y))^{-1} \geq 0$, we have 
by equation (\ref{eqn:newt-diff}) that 
$q^*- {\mathcal N}(y) \geq 0$, i.e., that $q^* \geq {\mathcal N}(y)$.
\end{proof} 

\noindent We also need the following, which is a less immediate consequence of
results in \cite{ESY12}:

\begin{lemma} \label{invnonneg} Given a PPS, $x= P(x)$, with LFP $q^* >0$,  if  $0 \leq y \leq q^*$, and $y < 1$, 
then $(I - P'(y))^{-1}$ exists and is non-negative. 
\end{lemma}
The proof of this lemma is more involved and is given in the appendix.
To prove the polynomial-time upper bounds in \cite{ESY12}, an inductive step of the 
following form was used:
\begin{lemma}[Combining Lemma 3.7 and Lemma 3.5 of \cite{ESY12}] 
\label{ppshalf} 
Let $x=P(x)$ be a PPS in SNF with $0 < q^* < 1$. For any
$0 \leq x \leq q^*$ and 
$\lambda > 0$, the operator $\mathcal{N}(x)$ is defined,
 $\mathcal{N}(x) \leq q^*$, and if
$q^* - x \leq \lambda(1-q^*)$
then
$q^* - \mathcal{N}(x) \leq \frac{\lambda}{2}(1-q^*)$.
\end{lemma}

If we knew an optimal policy $\tau$ for a max/minPPS, $x=P(x)$, 
then we would be able to solve the problem 
of computing the LFP for a max/minPPS using the algorithm in \cite{ESY12} 
for approximating $q^*_\tau$, 
because we know $q^*_\tau=q^*$.  
Unfortunately, we do not know which policy is optimal. 
There are exponentially many policies, so it would be inefficient to run 
this algorithm using every policy.   (And even if we did
do so for each possible policy, we would only be able
to $\epsilon$-approximate 
the values $q^*_{\sigma}$ for each policy $\sigma$ using the results of \cite{ESY12},
for say, $\epsilon = 2^{-j}$ for some chosen $j$, and therefore
we could only be sure  that a particular policy that
yields the best result is, say, $(2\epsilon)$-optimal, but it may not not 
necessarily be optimal.)
In fact, as we will see, it is probably impossible to identify an optimal
policy in polynomial time.

Our goal instead will be to find an iteration $I(x)$ for max/minPPS,
that has similar properties to the Newton operator for PPS,
i.e., that can be computed efficiently for a given $x$ and for which we can 
prove a similar property to Lemma \ref{ppshalf}, i.e. such that if 
$q^* - x \leq \lambda(1-q^*)$, then $q^* - I(x) \leq \frac{\lambda}{2}(1-q^*)$. 
Once we do so, we will be able to adapt 
and extend results from \cite{ESY12} to get a polynomial time 
algorithm for the problem of approximating the LFP $q^*$ of a max/minPPS.

\section{Generalizing Newton's method using linear programming}
 
If a max/minPPS, $x=P(x)$, has no equations of form Q, then it amounts
to precisely the Bellman equations for an ordinary finite-state Markov
Decision Process with the objective of maximizing/minimizing
reachability probabilities.  It is well known that we
can compute the exact (rational) optimal values for such finite-state 
MDPs, and thus the 
exact LFP, $q^*$, for such a max(min)-linear systems, using linear
programming (see, e.g., \cite{Puterman94,CY98}).  

Computing the LFP of max/minPPSs is clearly a generalization of this finite-state MDP
problem to the infinite-state setting of branching and recursive MDPs.
If
we have no equations of form M, we have a PPS, which we can solve
in P-time 
using Newton's method, as shown recently in \cite{ESY12}. 
An iteration of Newton's method works by approximating
the system of equations by a linear system. For a maxPPS(or minPPS),
we will define an analogous ``approximate'' system of equations
that we have to  solve in each iteration of
{\bf ``Generalized Newton's Method''} (GNM) which has both linear equations and
equations involving the max (or min) function. We will show
that we can solve the
equations that arise from each iteration of GNM using linear programming.
 We will then show that a polynomial (in fact, linear) number of iterations
are enough to approximate the desired LFP solution, and that it suffices
to carry out the computations with polynomial precision.

The rest of this Section is organized as follows.
In Section 3.1 we define a linearization of a max/minPPS and prove some
basic properties. 
In 3.2 we define the operator for an iteration of the
Generalized Newton's method and show that it can be computed by Linear Programming.
In Section 3.3 we analyze the operator for maxPPS and in Section 3.4 for minPPS.
Finally in Section 3.5 we put everything together and show that the algorithm approximates
the LFP within any desired precision in polynomial time in the Turing model.

\subsection{Linearizations of max/minPPSs and their properties}

We begin by expressing the max/min linear equations 
that should be solved by one iteration of 
what will eventually become the ``Generalized Newton's Method''  (GNM),
applied at a point $y$.  
Recall that we assume w.l.o.g. throughout that max/minPPS and PPS are in SNF.

\begin{definition}  For a max/minPPS, $x=P(x)$, with $n$ variables,
the {\bf linearization of $P(x)$ at a point ${\mathbf y} \in \mathbb{R}^n$},   
is a system of max/min linear functions denoted by $P^{y}(x)$, which has the
following form:

if $P(x)_i$ has form L or M,  then $P^y_i(x) = P_i(x)$, and

if $P(x)_i$ has form Q, i.e., $P(x)_i =x_jx_k$ for some $j$,$k$,
then $$P^y_i(x) = y_jx_k + x_jy_k - y_jy_k$$
\end{definition}

We can consider the linearization of a PPS, $x=P_\sigma(x)$, 
obtained as the result of fixing a policy, $\sigma$, for a 
max/minPPS, $x =P(x)$.

\begin{definition} $P^y_\sigma(x) := (P_\sigma)^y(x)$.\end{definition}

Note than the linearization $P^y(x)$ only changes equations of
form Q, and using a policy $\sigma$ only changes equations of form M, so these
operations are independent in terms of the effects they
have on the underlying equations, and thus
$P^y_\sigma(x) \equiv (P_\sigma)^y(x) = (P^y)_\sigma(x)$.

\begin{lemma} \label{lem:deriv-lin-orig-equal}
Let $x=P(x)$ be any PPS.
For any $y \in \real^n$,
let $(P^y)'(x)$ denote the Jacobian matrix of $P^y(x)$.  
Then for any $x \in \mathbb{R}^n$, we have $(P^y)'(x) = P'(y)$.
\end{lemma}
\begin{proof}
We need to show that the Jacobian
  $(P^y)'(x)$ of $P^y(x)$, evaluated anywhere, is equal to $P'(y)$.
  If $x_i = P_i(x)$ is not of form Q, then, for any $x \in
  \mathbb{R}^n$, $P_i(x) = P^y_i(x)$. So for any $x_j$,
  $\frac{\partial P^y_i(x)}{\partial x_j} = \frac{\partial
    P_i(x)}{\partial x_j}$.  Otherwise, $x_i=P_i(x)$ has form Q, that
  is $P_i(x) = x_jx_k$ for some variables $x_j$,$x_k$. Then $P^y_i(x)
  = y_jx_k + x_jy_k - y_jy_k$. In this case $\frac{\partial
    P^y_i(x)}{\partial x_j} =y_k$ and $\frac{\partial
    P^y_i(x)}{\partial x_k} =y_j$.  But when $x=y$, $\frac{\partial
    P_i(x)}{\partial x_j} =y_k$ and $\frac{\partial P_i(x)}{\partial
    x_k} =y_j$. Furthermore, clearly for any $x_l$, with $l \not= j$
  and $l \not= k$, $\frac{\partial P_i(x)}{\partial x_l} =0$ and
  $\frac{\partial P^y_i(x)}{\partial x_l} =0$. We have thus
  established that $(P^y)'(x) = P'(y)$ for any $x \in \mathbb{R}^n$.
\end{proof}

\begin{lemma}\label{lem:lin-form-deriv-equation} 
If $x=P(x)$ is any PPS, 
then for any $x,y \in \mathbb{R}^n$,
 $P^y(x) = P(y) + P'(y)(x-y)$.
\end{lemma}
\begin{proof} 
  Firstly, note that $P^y(x) = P^y(y) + (P^y)'(x)(x-y)$, since the
  functions $P^y_i(x)$ are all linear in $x$.  Next, observe that
  $P_i(y) = P^y_i(y)$, for all $i$, and thus that $P(y) = P^y(y)$.
  Thus, to show that $P^y(x) = P^y(y) + P'(y)(x-y) = P(y) +
  P'(y)(x-y)$, all we need to show is that the Jacobian
  $(P^y)'(x)$ of $P^y(x)$, evaluated anywhere, is equal to $P'(y)$.
But this was established in Lemma \ref{lem:deriv-lin-orig-equal}.
\end{proof}

An iteration of Newton's method on $x= P_\sigma(x)$ at a point $y$
solves a system of linear equations that can be expressed in terms of
$P^y_\sigma(x)$.  
The next lemma establishes this basic fact in part {\em (i)}. 
In part {\em (ii)} it provides us with conditions 
under which we are guaranteed to be doing ``at least as well'' as one such
Newton iteration.

\begin{lemma} \label{fpstrat} Suppose that the matrix inverse 
$(I-P'_\sigma(y))^{-1}$
  exists and is non-negative, for some policy $\sigma$, 
and some $y \in  \mathbb{R}^n$. Then
\begin{itemize}
\item[(i)] $\mathcal{N}_\sigma(y)$ is defined, and is equal to the unique point $a \in \real^n$ 
such that  $P^y_\sigma(a) = a$.
\item[(ii)] For any vector $x \in \real^n$:\\
If $P^y_\sigma(x) \geq x$, then $x \leq \mathcal{N}_\sigma(y)$.\\
If $P^y_\sigma(x) \leq x$, then $x \geq \mathcal{N}_\sigma(y)$.
\end{itemize}
\end{lemma}
\begin{proof}
(i): We define:
$$a = y +(I-P'_\sigma(y))^{-1}(P_\sigma(y) - y) \equiv {\mathcal{N}}_\sigma(y)$$
Then we can re-arrange this expression, reversibly, yielding:
\begin{eqnarray*}
a = y +(I-P'_\sigma(y))^{-1}(P_\sigma(y) - y) & \Leftrightarrow & P_\sigma(y) - y - (I-P'_\sigma(y))(a-y) = 0\\
& \Leftrightarrow & P_\sigma(y) + P'_\sigma(y)(a-y) = a\\
& \Leftrightarrow &  P^y_\sigma(a) = a \quad \quad \mbox{(by Lemma \ref{lem:lin-form-deriv-equation})}
\end{eqnarray*}
Uniqueness follows from the reversibility of these transformations.\\
 
\noindent (ii): Firstly, we shall observe that the result of applying Newton's method to solve 
$x = P^y_\sigma(x)$ with any initial point $x$ gives us $\mathcal{N}_\sigma(y) = a$ in a single iteration.  Recalling
from Lemma \ref{lem:deriv-lin-orig-equal} that the following equality hold
between the Jacobians: $(P^y)'(x) = P'_\sigma(y)$,  one iteration of Newton's method applied 
to $x = P^y_\sigma(x)$ can be equivalently defined as:
\begin{eqnarray*}x + (I-P'_\sigma(y))^{-1}(P^y_\sigma(x)-x) & = & x + (I-P'_\sigma(y))^{-1}(P_\sigma(y) + P'_\sigma(y)(x-y) -x)\\
											& = & (I-P'_\sigma(y))^{-1}( x - P'_\sigma(y)x + P_\sigma(y) + P'_\sigma(y)(x-y) -x)\\
											& = & (I-P'_\sigma(y))^{-1}(P_\sigma(y) - P'_\sigma(y)y)\\
											& = & (I-P'_\sigma(y))^{-1}((I-P'_\sigma(y))y + P_\sigma(y) -y)\\
											& = & y + (I-P'_\sigma(y))^{-1}(P_\sigma(y) - y)\\											
& = & \mathcal{N}_\sigma(y).\end{eqnarray*}
We thus have $\mathcal{N}_\sigma(y)= x + (I-P'_\sigma(y))^{-1}(P^y_\sigma(x)-x)$.  By assumption, 
$(I-P'_\sigma(y))^{-1}$  is a non-negative matrix. So if $P^y_\sigma(x)-x \geq 0$ then $\mathcal{N}_\sigma(y)\geq x$,
whereas if $P^y_\sigma(x)-x \leq 0$ then $\mathcal{N}_\sigma(y)\leq x$.\end{proof}

\subsection{The iteration operator of Generalized Newton's Method}

We shall now 
define distinct iteration operators for a maxPPS and a minPPS,
both of which we shall refer to with the overloaded notation $I(x)$. 
(We shall also establish in the next two subsections that the operators are 
 well-defined
 in their respective settings.)
These operators will serve as the basis for a
{\em Generalized Newton's Method} to be applied to maxPPSs and minPPSs, respectively.

\begin{definition}
For a maxPPS, $x=P(x)$, with LFP $q^*$, such that $0 < q^* < 1$,
and for a real vector $y$ such
that $0 \leq y \leq q^*$, we define the operator $I(y)$ to be the 
{\em unique} optimal solution, $a \in \real^n$, to the following 
mathematical program: $\quad \quad \mbox{\em Minimize:} \  \ \sum_i a_i \ ; \quad  
 \mbox{\em Subject to:} \quad   P^y(a) \leq a$.

For a minPPS, $x=P(x)$, with LFP $q^*$, such that $0 < q^* < 1$, 
and for a real vector $y$ such 
that $0 \leq y \leq q^*$, we define the operator $I(y)$ to be the 
{\em unique} optimal solution  
$a \in \real^n$ 
to the following mathematical program:
$ \quad \quad \mbox{\em Maximize:} \  \ \sum_i a_i \ ; \quad 
 \mbox{\em Subject to:} \quad   P^y(a) \geq a$.

\end{definition}

\noindent
A priori, it is not even clear if the above ``definitions'' of $I(x)$ for
maxPPSs and minPPSs are well-defined.
We now make the following central claim, which we shall prove 
separately for maxPPSs and minPPSs 
in the following two subsections:

\begin{proposition} \label{subgoal} Let  $x=P(x)$ be a max/minPPS, with LFP  $q^*$, such that $0 < q^* < 1$.
For any $0 \leq x \leq q^*$:
\begin{enumerate}
\item $I(x)$ is well-defined,  and $I(x) \leq q^*$, and:
\item For any $\lambda > 0$,  if
$q^* - x \leq \lambda(1-q^*)$
then
$q^* - I(x) \leq \frac{\lambda}{2}(1-q^*)$.
\end{enumerate}
\end{proposition}

The next proposition observes that linear programming can be used to compute 
an iteration of the operator, $I(x)$, for both maxPPSs and minPPSs.

\begin{proposition}\label{prop:comp-I} 
Given a max/minPPS, $x=P(x)$, with LFP $q^*$, and given a rational vector $y$, $0 \leq y \leq q^*$,
the constrained optimization problem (i.e., mathematical program) ``defining'' $I(y)$ 
can be described by a LP 
whose encoding size is polynomial (in fact, linear) in both $|P|$ and
the encoding size of the rational vector $y$.
Thus, we can compute the (unique) optimal solution $I(y)$ to such an LP (assuming it exists, and is unique) in P-time.
\end{proposition}
\begin{proof}
For a maxPPS (minPPS), the definition of $I(x)$ asks us to maximize (minimize) a linear objective, 
$\sum_i a_i$, subject to the
constraints $P^y(a) \leq a$  ($P^y(a) \geq a$, respectively).
All of these constraints are linear, except the constraints of form M.
For a maxPPS,
if $(P^y(a))_i$ is of form M, then the corresponding constraint is an inequality of the form 
$\text{max } \{a_j,a_k\} \leq a_i$.   Such an inequality is equivalent to, and can be replaced by, 
the two linear inequalities:
$a_j \leq a_i$ and $a_k \leq a_i$. 
Likewise,  
for a minPPS, if  $(P^y(a))_i$ is of form M,  
then the corresponding constraint is an inequality of the form $\text{min } \{a_j,a_k\} \geq a_i$. 
Again, such an inequality is equivalent to, and can be replaced by, two linear inequalities:
$a_j \geq a_i$ and $a_k \geq a_i$.

Thus, for a rational vector $y$ whose encoding length is $\size(y)$, the operator $I(y)$ can be formulated 
(for both maxPPSs and minPPSs)
as a problem of
computing the unique optimal solution to 
a linear program
whose encoding size is polynomial (in fact, linear) in $|P|$ and in $\size(y)$.
\end{proof}

\subsection{An iteration of Generalized Newton's Method (GNM) for maxPPSs}

For a maxPPS, $x=P(x)$, we know by Theorem \ref{optexist} 
that there exists an optimal policy, $\tau$, such that $q^*=q^*_\tau \geq q^*_{\sigma}$ for any
policy $\sigma$.  The next lemma implies part (i) of Proposition 
\ref{subgoal} for maxPPS:

\begin{lemma} \label{itworks} If $x =P(x)$ is a maxPPS, with LFP solution $0 < q^* < 1$, 
and $y$ is a real vector with $0 \leq y \leq q^*$, 
then $x=P^y(x)$ has a least fixed point solution, denoted $\mu P^y$, with $\mu P^y \leq q^*$.
Furthermore, the operator $I(y)$ is well-defined, $I(y) = \mu P^y \leq q^*$, 
and for any 
optimal policy $\tau$, $I(y) = \mu P^y \geq \mathcal{N}_\tau(y)$.
\end{lemma}
\begin{proof}
Recall that (by Proposition \ref{prop:comp-I}) the following can be written as an LP that ``defines''  $I(y)$:

\begin{equation} \label{eq:lp-max}
 \mbox{\em Minimize:} \  \ \sum_i a_i \ ; \quad \quad \quad
 \mbox{\em Subject to:} \quad   P^y(a) \leq a
\end{equation}

Firstly, we show that the LP 
constraints $P^y(a) \leq a$ in the definition of $I(y)$ are {\em feasible}.
We do so by showing that actually $P^y(q^*) \leq q^*$. 
At any coordinate $i$, if $P_i(x)$ has form M or L, then $P^y_i(q^*) = P_i(q^*)=q^*_i$. Otherwise, $P_i(x)$
has form $Q$, i.e.,  $P_i(x)=x_jx_k$, and then 
\begin{eqnarray*} P^y_i(q^*) & = & q^*_jy_k + y_jq^*_k - y_jy_k\\
							& = & q^*_jq^*_k - (q^*_j - y_j)(q^*_k - y_k)\\
							& \leq & q^*_i    \quad  \quad \quad \quad \mbox{(since $y \leq q^*$)}
\end{eqnarray*}

Next we show that the LP (\ref{eq:lp-max}) defining $I(y)$ is {\em bounded}.
Recall that, by Theorem \ref{optexist}, 
there is always an optimal policy for any maxPPS, $x=P(x)$.

\begin{claim} \label{betterthanoptimal} 
Let $x=P(x)$ be any maxPPS, with $0 < q^* < 1$,  and let $\tau$ be any optimal policy for $x=P(x)$. 
For any $y$ such that $0 \leq y \leq q^*$, we have that ${\mathcal{N}}_{\tau}(y)$ is defined, 
and for any vector $a$,
if $P^y(a) \leq a$ then $\mathcal{N}_\tau(y) \leq a$.
In particular, ${\mathcal{N}}_{\tau}(y) \leq q^*$.
\end{claim}
\begin{proof}  Recall, from our definition of an optimal policy, that $q^* = q^*_{\tau}$ is also the least 
non-negative solution to $x=P_\tau(x)$. So we can apply Lemma \ref{invnonneg} using $x=P_\tau(x)$ and $y \leq q^*$ 
to deduce that $(I-P'_\tau(y))^{-1}$ exist and is non-negative. 
Thus  ${\mathcal{N}}_{\tau}(y)$ is defined. 
Now, by applying Lemma \ref{fpstrat} (ii),  to show 
that $a \geq \mathcal{N}_\tau(y)$ all we need to show is that $P^y_\tau(a) \leq a$. 
But recalling that $x=P(x)$ is a maxPPS, by the definition of $P^y(x)$ and $P^y_\tau(x)$, 
we have that $P^y_\tau(a) \leq P^y(a) \leq a$.
We have just shown before this Claim that $P^y(q^*) \leq q^*$, and thus 
$\mathcal{N}_\tau(y) \leq q^*$.
\end{proof}

Thus the LP (\ref{eq:lp-max}) defining $I(y)$ is both feasible and bounded,
hence it has an optimal solution.
To show that $I(y)$ is well-defined, 
all that remains is to show that this
optimal solution is unique.  In the process, we will also show that $I(y)$ defines
precisely the {\em least fixed point} solution of $x=P^y(x)$, which we denote by $\mu P^y$.

Firstly, we claim that for any optimal solution $b$ to the LP (\ref{eq:lp-max}),  it must
be the case that $P^y(b) = b$.    Suppose not. Then there exists $i$ such that $P^y(b)_i < b_i$, 
then we can define a new vector $b'$, such that $b'_i = P^y(b)_i$ and $b'_j = b_j$ for all $j \neq i$.
By monotonicity of $P^y(x)$, it is clear that $P^y(b') \leq b'$, and thus that $b'$ is a
feasible solution to the LP (\ref{eq:lp-max}).  
But $\sum_i b'_i < \sum_i b_i$, contradicting the assumption that $b$ is an optimal solution to the LP (\ref{eq:lp-max}).

Secondly, we claim that there is a unique optimal solution.  Suppose not:  suppose $b$ and $c$ 
are two distinct optimal solution to the LP (\ref{eq:lp-max}).
Define a new vector $d$  by  $d_i = \min \{ b_i, c_i\}$, for all $i$.
Clearly, $d \leq b$ and $d \leq c$.
Thus by the monotonicity of $P^y(x)$, for all $i$ $P^y(d)_i \leq P^y(b)_i = b_i$,
and likewise $P^y(d)_i \leq P^y(c)_i = c_i$.  Thus $P^y(d) \leq d$, and $d$ is a feasible solution to the LP. 
But since $b$ and $c$ are distinct, and yet $\sum_i b_i = \sum_i c_i$, we have
that $\sum_i d_i < \sum_i b_i = \sum_i c_i$, contradicting the optimality of both $b$ and $c$.

We have thus established that $I(y)$ defines the unique {\em least fixed point} solution 
of $x=P^y(x)$, which we denote also by $\mu P^y$.
Since $q^*$ is also a solution of the LP, we have $\mu P^y \leq q^*$.

Finally, by Claim \ref{betterthanoptimal}, it must be the case that
$I(y) = \mu P^y \geq {\mathcal{N}}_{\tau}(y)$, where $\tau$ is any optimal
policy for $x=P(x)$.
\end{proof}

We next establish part (ii) of Proposition \ref{subgoal} for maxPPS.

\begin{lemma} \label{maxhalf} Let $x=P(x)$ be a maxPPS with $0 < q^* < 1$.  
For any $0 \leq x \leq q^*$ and $\lambda > 0$, we have
$I(x) \leq q^*$, and furthermore if:
$$q^* - x \leq \lambda(1-q^*)$$
then
$$q^* - I(x) \leq \frac{\lambda}{2}(1-q^*)$$
\end{lemma}
\begin{proof} 
Let $\tau$ be an optimal policy (which exists by Theorem \ref{optexist}). 
The least fixed point solution of the PPS $x=P_\tau(x)$ is $q^*$. 
From our assumptions, Lemma \ref{ppshalf} gives that $q^* - \mathcal{N}_\tau(x) \leq \frac{\lambda}{2}(1-q^*)$. 
But by Lemma \ref{itworks} 
$\mathcal{N}_\tau(x) \leq I(x) \leq q^*$.
The claim follows.
\end{proof}

Proposition \ref{subgoal} for maxPPSs 
follows from Lemmas \ref{itworks} and \ref{maxhalf}.
In subsection \ref{atlast} we will combine this result with methods from \cite{ESY12} to obtain a P-time algorithm 
for approximating the LFP of a maxPPS,
in the standard Turing model of computation.

\subsection{An iteration of GNM for minPPSs}

Our proof of the minPPS version of Lemma \ref{maxhalf} will be 
somewhat different, because it
turns out we can not use the same argument based on LPs 
to prove that $I(y)$ is well-defined. 
Fortunately, in the case of minPPSs, we can show 
that $(I - P_\sigma(y))^{-1}$ exists and is non-negative for {\em any}
policies $\sigma$, at those points $y$ that are of interest. 
And we can use this to show that there is
{\em some} policy, $\sigma$, such that
$I(y)$ is equivalent to an iteration of Newton's method at 
$y$ after fixing the policy $\sigma$.
We shall establish the existence of such a policy using a
policy improvement argument, instead of just using the LP, 
as we did for maxPPSs.
(Note that the policy improvement algorithm may not
 be an efficient (P-time) way to compute it, and we do not claim it is.
 We only use policy improvement as an argument in the proof of existence
of a suitable policy $\sigma$.)

\begin{lemma} \label{qminimal} 
For a minPPS, $x=P(x)$, and for any policy $\sigma$, the 
LFP of, $x=P_\sigma(x)$, denoted $q^*_\sigma$, 
satisfies $q^* \leq q^*_\sigma$.\end{lemma}
\begin{proof} By Theorem \ref{optexist}, 
there is an optimal policy $\tau$ with $q^*_\tau=q^*$. But we 
defined an optimal policy for a minPPS as one with $q^*_\tau \leq q^*_\upsilon$ for any policies $\upsilon$. So $q^* = q^*_\tau \leq q^*_\sigma$. \end{proof}

\noindent Lemma \ref{qminimal} allows us to use 
Lemma \ref{invnonneg} with any policy, not just with optimal policies:

\begin{lemma}\label{lem:min-all-strat-defined} For a minPPS, $x=P(x)$, with LFP $0 < q^* <1$, 
for any $0 \leq y \leq q^*$ and any policy 
$\sigma$, $(I - P_\sigma(y))^{-1}$ exists and is non-negative.
Thus also  ${\mathcal N}_{\sigma}(y)$ is defined.
\end{lemma}

\begin{proof} $0 \leq y \leq q^* \leq q^*_\sigma \leq 1$. 
Note also that $y < 1$, and that $q^*_\sigma \geq  q^* > 0$. 
This is all we need for Lemma \ref{invnonneg} to apply.
\end{proof}

\begin{lemma} \label{minattained} Given a minPPS, $x=P(x)$, 
with LFP $0 < q^* <1$, and a vector $y$ with $0 \leq y \leq q^*$, 
there is a policy $\sigma$ such that 
$P^y(\mathcal{N}_\sigma(y)) = \mathcal{N}_\sigma(y)$.
\end{lemma}
\begin{proof} We use a policy (strategy) improvement ``algorithm'' to prove this. 
Start with any policy $\sigma_1$. At step $i$, 
suppose we have a policy $\sigma_i$.

For notational simplicity, in the
following we use the abbreviation:  $z =\mathcal{N}_{\sigma_i}(y)$. 
By Lemma \ref{fpstrat}, $P^y_{\sigma_i}(z) = z$. 
So we have $P^y(z) \leq z$. 
If $P^y(z) = z$, then {\em stop}: we are done.

Otherwise, to construct the next strategy $\sigma_{i+1}$,
take the smallest $j$ such that $(P^y(z))_j < z_j$. 
Note that $P_j(x)$ has form M, because
otherwise $(P(x))_j = (P_{\sigma_i}(x))_j$.
Thus, there is some variable $x_k$ with 
$P_j(x) = \text{min }\{ x_k, x_{\sigma_i(j)} \}$ and 
$z_k < z_{\sigma_i(j)}$. Define $\sigma_{i+1}$ to be:
$$\sigma_{i+1}(l) = \begin{cases}
\sigma_i(l) & \text{if } l \not= j \\
 k & \text{if } l = j
 \end{cases}$$
 Then $(P^y_{\sigma_{i+1}}(z))_j < z_j$, but for every other coordinate $l \not= j$,
 $(P^y_{\sigma_{i+1}}(z))_l = (P^y_{\sigma_i}(z))_l = z_l$. Thus 
\begin{equation}\label{eq:proof-improve-min-ineq}
P^y_{\sigma_{i+1}}(z) \leq z
\end{equation}
\noindent By 
Lemma 
\ref{lem:min-all-strat-defined},
$\mathcal{N}_{\sigma_{i+1}}(y)$ is defined.
Moreover, the inequality (\ref{eq:proof-improve-min-ineq}), together with
 Lemma \ref{fpstrat} (ii), yields that 
$\mathcal{N}_{\sigma_{i+1}}(y) \leq z$.  
But $\mathcal{N}_{\sigma_{i+1}}(y) \neq z$ because 
$P^y_{\sigma_{i+1}}(z) \not= z$ whereas, by Lemma \ref{fpstrat} (i), 
we have 
$P^y_{\sigma_{i+1}}(\mathcal{N}_{\sigma_{i+1}}(y)) = \mathcal{N}_{\sigma_{i+1}}(y)$.

Thus this algorithm gives us a sequence of policies 
$\sigma_1, \sigma_2 ...$ with $\mathcal{N}_{\sigma_1}(y) \geq \mathcal{N}_{\sigma_2}(y) \geq \mathcal{N}_{\sigma_3}(y) \geq ...$, 
where furthermore each step must strictly decrease at least one 
coordinate of $\mathcal{N}_{\sigma_i}(y)$. 
It follows that $\sigma_i \not= \sigma_j$, unless $i=j$. 
There are only finitely many policies. So the sequence must be finite, and  
the algorithm terminates. But it only terminates when we reach a 
$\sigma_i$ with $P^y(\mathcal{N}_{\sigma_i}(y)) = \mathcal{N}_{\sigma_i}(y)$.
 \end{proof}
 
\noindent We note that the analogous policy improvement algorithm might fail
to work for maxPPSs, as we might reach a policy $\sigma_i$ where $(I -
 P_{\sigma_i}(x))^{-1}$ does not exist, or has a negative entry.

The next Lemma shows that
this policy improvement algorithm always produces 
a coordinate-wise minimal Newton iterate over all policies.

\begin{lemma} \label{ufpmin} 
For a minPPS, $x = P(x)$, with LFP $0 < q^* < 1$,
if $0 \leq y \leq q^*$ and $\sigma$ is a policy such that 
$P^y(\mathcal{N}_\sigma(y)) = \mathcal{N}_\sigma(y)$, then:
\begin{itemize}
\vitem[(i)] For any policy $\sigma'$, $\mathcal{N}_{\sigma'}(y) \geq \mathcal{N}_{\sigma}(y)$.
\vitem[(ii)]For any $x \in \real^n$ with $P^y(x) \geq x$, we have $x  \leq \mathcal{N}_\sigma(y)$.
\vitem[(iii)] For any $x \in \real^n$ with $P^y(x) \leq x$, we have $x  \geq \mathcal{N}_\sigma(y)$.
\vitem[(iv)]$\mathcal{N}_\sigma(y)$ is the unique fixed point of $x = P^y(x)$.
\vitem[(v)] $\mathcal{N}_\sigma(y) \leq q^*$.
\end{itemize}\end{lemma}
\begin{proof} Note firstly that by Lemma \ref{lem:min-all-strat-defined}, 
for any policy $\sigma$, 
$(I-P'_\sigma(y))^{-1}$ exists and is non-negative, and ${\mathcal N}_\sigma(y)$ 
is defined. 
\begin{itemize}
\vitem[(i)]  Consider $P^y_{\sigma'}(\mathcal{N}_\sigma(y))$.  Note that 
$P^y_{\sigma'}(\mathcal{N}_\sigma(y)) \geq P^y(\mathcal{N}_\sigma(y)) 
= \mathcal{N}_\sigma(y)$ by assumption. 
Thus, by Lemma \ref{fpstrat} (ii), 
$\mathcal{N}_{\sigma}(y) \leq \mathcal{N}_{\sigma'}(y)$.
\vitem[(ii)] $P^y_\sigma(x) \geq P^y(x) \geq x$, so by Lemma \ref{fpstrat} (ii), 
$x  \leq \mathcal{N}_\sigma(y)$.
\vitem[(iii)] If $P^y(x) \leq x$, then there a policy $\sigma'$ with 
$P^y_{\sigma'}(x) \leq x$, and by Lemma \ref{fpstrat} (ii), $x \geq \mathcal{N}_{\sigma'}(y)$. So using part (i)
of this Lemma, $x \geq \mathcal{N}_{\sigma'}(y) \geq  \mathcal{N}_{\sigma}(y)$.
\vitem[(iv)] By assumption, $\mathcal{N}_\sigma(y)$ is a fixed point of 
$x=P^y(x)$. We just need uniqueness. If $P^y(q) = q$, then by parts (ii) and 
(iii) of this Lemma, $q \leq  \mathcal{N}_\sigma(y)$ and  $q \geq  \mathcal{N}_\sigma(y)$, i.e., $q = \mathcal{N}_\sigma(y)$.
\vitem[(v)] 
Consider an 
optimal policy $\tau$, for the minPPS, $x=P(x)$.
From Lemma \ref{newtbounds},  if follows that 
$\mathcal{N}_\tau(y) \leq q^*_\tau = q^*$. 
And then part 
(i) of this Lemma, gives us that
$\mathcal{N}_\sigma(y) \leq \mathcal{N}_\tau(y) \leq q^*$. 
\end{itemize}\end{proof}

\noindent We can now return to using linear programming, which we can do
in polynomial time.  Recall the LP that ``defines'' $I(y)$, for a minPPS:
\begin{equation}
\label{eq:lp-min}
 \mbox{\em Maximize:} \  \ \sum_i a_i \ ; \quad \quad \quad
 \mbox{\em Subject to:} \quad   P^y(a) \geq a
\end{equation}

\begin{lemma}
\label{lem:min-unique-I(y)} For a minPPS, $x=P(x)$, with LFP $0 < q^* < 1$,
and for
$0 \leq y \leq q^*$, there is a unique optimal solution, which we call $I(y)$, 
to the LP (\ref{eq:lp-min}),
and furthermore $I(y) = \mathcal{N}_\sigma(y)$ for some policy $\sigma$, 
and $P^y(I(y)) = I(y)$. \end{lemma}
\begin{proof}
By Lemma \ref{minattained}, 
there is a $\sigma$ such that $P^y(\mathcal{N}_\sigma(y)) = \mathcal{N}_\sigma(y)$. So $\mathcal{N}_\sigma(y)$ is a feasible solution of $P^y(a) \geq a$. Let $a$ by any solution of $P^y(a) \geq a$. By Lemma \ref{ufpmin} (ii), $a \leq \mathcal{N}_\sigma(y)$. Consequently $\sum_{i=1}^n a_i \leq  \sum_{i=1}^n (\mathcal{N}_\sigma(y))_i$ with equality only if $a= \mathcal{N}_\sigma(y)$. So $\mathcal{N}_\sigma(y)$ is the unique optimal solution of the LP (\ref{eq:lp-min}).
\end{proof}

\noindent In the maxPPS case, we had an iteration that was at least as good as iterating with the optimal policy. Here we have an iteration that is at least as 
bad!  Nevertheless, we shall see that it is good enough.
In the maxPPS case, the analog of Lemma \ref{ppshalf}, 
Lemma \ref{maxhalf}, thus followed from Lemma \ref{ppshalf}. Here we 
crucially need a stronger result than Lemma \ref{ppshalf}.

\begin{lemma} \label{halfineq} If $x=P(x)$ is a PPS and we are given $x,y \in \mathbb{R}^n$ with $0 \leq x \leq y \leq P(y) \leq 1$, and if 
the following conditions hold:
\begin{equation}
\label{eq:assumptions-long-derivation-min}
\lambda>0 \quad  \mbox{and} \quad   y-x \leq \lambda (\textbf{1}-y) 
\quad \mbox{and}
\quad (I - P'(x))^{-1} \ \mbox{exists and is non-negative,}
\end{equation}
then $y - \mathcal{N}(x) \leq \frac{\lambda}{2} (\textbf{1}-y)$. \end{lemma}
\noindent (Note that we cannot conclude that $y - \mathcal{N}(x) \geq 0$.)
\begin{proof} 
Firstly, we show that $P'(y)(\textbf{1}-y) \leq (\textbf{1}-y)$. 
Clearly, for any PPS, $P(\textbf{1}) \leq  1$.   
Note  that since by assumption $y \leq P(y)$,
we have $(\textbf{1}-y) \geq  (\textbf{1}-P(y)) \geq (P(\textbf{1}) - P(y))$.
Then by Lemma 3.3 of \cite{ESY12}:
 \begin{eqnarray}
(\textbf{1}- y) \geq P(\textbf{1}) - P(y) & = & P'(\frac{\textbf{1}+y}{2})(\textbf{1}-y) \\
					& \geq & P'(y)(\textbf{1}-y) 
\label{eqar:1-y-cone-min}\end{eqnarray}
Again by Lemma 3.3 of \cite{ESY12}: $P(y) - P(x) = \frac{1}{2}(P'(x) + P'(y))(y-x)$, and thus:
\begin{equation}
\label{eq:used-inside-new-halfing}
P(x) = P(y) - \frac{1}{2}(P'(x) + P'(y))(y-x) 
\end{equation}
Thus:
\begin{eqnarray*}
  y - {\mathcal N}(x) & = & y - x - (I-P'(x))^{-1}(P(x) - x)\\
 & = & y - x - (I-P'(x))^{-1}(P(y) - x - \frac{1}{2}(P'(x) + P'(y))(y-x))  \quad \mbox{(by (\ref{eq:used-inside-new-halfing}))}\\
 & \leq & y - x - (I-P'(x))^{-1}(y - x - \frac{1}{2}(P'(x) + P'(y))(y-x)) \\
 & = & (y - x) - (I-P'(x))^{-1}( (y - x)  - \frac{1}{2}(P'(x) + P'(y))(y-x)) \\
  & = & (I - (I-P'(x))^{-1}( I  - \frac{1}{2}(P'(x) + P'(y)))) (y-x)\\
  & = & ((I-P'(x))^{-1}(I-P'(x)) - (I-P'(x))^{-1}( I  - \frac{1}{2}(P'(x) + P'(y)))) (y-x)\\
 & = & (I-P'(x))^{-1} (I-P'(x) -  (I  - \frac{1}{2}(P'(x) + P'(y)))) (y-x)\\
 & = & (I-P'(x))^{-1} (-P'(x) + \frac{1}{2}(P'(x) + P'(y))) (y-x)\\
  & = & (I-P'(x))^{-1}\frac{1}{2}(P'(y) - P'(x))(y-x) \\
  & \leq & \frac{\lambda}{2}(I-P'(x))^{-1}(P'(y) - P'(x))(\textbf{1}-y) \quad
\mbox{(by (\ref{eq:assumptions-long-derivation-min}), and because 
$(P'(y) - P'(x)) \geq 0$)}\\
  & \leq & \frac{\lambda}{2}(I-P'(x))^{-1}(I - P'(x))(\textbf{1}-y) \quad
\mbox{(because by (\ref{eqar:1-y-cone-min}),  $P'(y)(1-y) \leq (1-y)$)}\\
  & = & \frac{\lambda}{2}(\textbf{1}-y) \end{eqnarray*}
\end{proof}

\begin{lemma} \label{minhalf} Let $x=P(x)$ be a minPPS, with LFP $0 < q^* < 1$. For any $0 \leq x \leq q^*$  and $\lambda > 0$, 
$I(x) \leq q^*$, and if:
$$q^* - x \leq \lambda(1-q^*)$$
then
$$q^* - I(x)\leq \frac{\lambda}{2}(1-q^*)$$
\end{lemma}
\begin{proof} 
By Lemma \ref{minattained}, there is a policy $\sigma$ with 
$I(x) = \mathcal{N}_\sigma(x)$.  We then apply Lemma \ref{halfineq} to 
$x= P_\sigma(x)$, $x$, and $q^*$ instead of $y$. Observe that $P_\sigma(q^*) \geq P(q^*) = q^*$ and that $(I-P'_\sigma(x))^{-1}$ exists and is non-negative. 
Thus the conditions of Lemma \ref{halfineq} hold, and we can 
conclude that $q^* - \mathcal{N}_\sigma(x) \leq \frac{\lambda}{2}(1-q^*)$.
Lastly, Lemma \ref{ufpmin} (v) 
and Lemma \ref{lem:min-unique-I(y)} yield that $I(x) = \mathcal{N}_\sigma(x) \leq q^*$.
\end{proof}

\noindent Proposition \ref{subgoal} for minPPS follows  from Lemmas \ref{lem:min-unique-I(y)} and \ref{minhalf}.

\subsection{\label{atlast} A polynomial-time algorithm (in the Turing model) for max/minPPSs}

In \cite{ESY12} we gave a polynomial time algorithm, in the 
standard Turing model
of computation,
for approximating the LFP of a PPS, $x=P(x)$, using Newton's method. 
Here we use the same methods from \cite{ESY12}, with  our new 
{\em Generalized Newton's Method} (GNM), $I(x)$, to obtain polynomial-time algorithms (again, in the standard Turing model), for approximating the LFP of maxPPSs and minPPSs.
The proof in \cite{ESY12} uses induction based on the ``halving lemma'', 
Lemma \ref{ppshalf}.  
We of course now have suitable ``halving lemmas'' for maxPPSs and minPPSs,
namely, Lemmas \ref{maxhalf} and \ref{minhalf}.
In \cite{ESY12}, the following bound was used
for the base case of the induction:
\begin{lemma}[Theorem 3.12 from \cite{ESY12}] \label{1-qbound}
If $0 < q^* < 1$ is the LFP of a PPS, $x = P(x)$, in $n$ variables, 
then for all $i \in \{1,\ldots, n\}$:
$$1 - q^*_i \geq 2^{-4|P|}$$
In other words, $0 < q^*_i \leq 1 - 2^{-4|P|}$
, for all $i \in \{1, \ldots,n\}$.\end{lemma}

\noindent We can now easily derive an analogous Lemma for  
the setting of max/minPPSs:

\begin{lemma} \label{gap} If $0 < q^* < 1$ is the LFP of a max/minPPS, $x = P(x)$, 
in $n$ variables, then for all $i \in \{1, \ldots , n\}$:
$$1 - q^*_i \geq 2^{-4|P|}$$
In other words, $0 < q^*_i \leq 1 - 2^{-4|P|}$,
for all $i \in  \{1, \ldots, n\}$.
\end{lemma}
\begin{proof} Let $\tau$ be any optimal policy
for $x=P(x)$.  We know it exists, by Theorem
\ref{optexist}.
Lemma \ref{1-qbound} gives that $1 - q^*_i \geq 2^{-4|P_\tau|}$. 
All we need is to note is that $|P| \geq |P_\tau|$,
which clearly holds using any sensible encoding for $P$ and $P_\tau$,
in the sense that we should need no more bits needed to 
encode  $x_i = x_j$ than to encode $x_i = \max \{x_j,x_k\}$ or 
$x_i = \min \{x_j,x_k\}$.
\end{proof}

Now we can give a polynomial time algorithm, in the Turing model
of computation,
for approximating the LFP, $q^*$, for a max/minPPS, 
to within any desired precision, 
by carrying out iterations of GNM using the same rounding technique, 
with the same rounding parameter, 
and using the same number of iterations, as in \cite{ESY12}.
Specifically, we use the following algorithm with rounding parameter $h$:

\medskip
\noindent Start with $x^{(0)} :=0$;\\
For each $k \geq 0$ compute $x^{(k+1)}$ from $x^{(k)}$ as follows:
\begin{enumerate}
\item Calculate $I(x^{(k)})$  by solving the following LP:\\
{\em Minimize:} \ $\sum_i x_i$ \ ; \ {\em Subject to:} $P^{x^{(k)}}(x) \leq x$,  
if $x=P(x)$ is a maxPPS,\\ 
 or: \\
{\em Maximize:} \ $\sum_i x_i$\ ; \   {\em Subject to:} $P^{x^{(k)}}(x) \geq x$, 
if $x=P(x)$ is a minPPS.

\item For each coordinate $i=1,2,...n$, set $x_i^{(k+1)}$ to be the 
maximum (non-negative) multiple of $2^{-h}$ which is $\leq 
\text{max}\{0,I(x^{(k)})_i\}$. 
(In other words, we round $I(x^{(k)})$ down to the nearest $2^{-h}$ and ensure it is non-negative.)
\end{enumerate}

\begin{theorem} \label{mptime} Given any max/minPPS, $x=P(x)$, with 
LFP $0 < q^* < 1$,  if we use
the above algorithm with rounding parameter $h = j +2+4|P|$, then the iterations are all defined,
and for every $k \geq 0$ we have $0 \leq x^{(k)} \leq q^*$, and furthermore after $h = j + 2 + 4|P|$ iterations we
have:
$$\|q^*- x^{(j+2+4|P|)}\|_\infty \leq 2^{-j}$$
\end{theorem}

\noindent The proof is very similar to the proof of Theorem 4.2 
in \cite{ESY12},
and is given in the Appendix.

\begin{corollary} \label{mcor} Given any max/minPPS, $x=P(x)$, 
with LFP $q^*$, and given any integer $j>0$, there is an algorithm that
computes a rational vector $v$ with $\|q^* -v\|_\infty \leq 2^{-j}$,
in time polynomial in $|P|$ and $j$.\end{corollary}            
\begin{proof} First, we use the algorithms 
given in \cite{rmdp} (Theorems 11 and 13), to detect 
those variables $x_i$ with $q^*_i =0$ or $q^*_i = 1$ in time polynomial in 
$|P|$. Then we can remove these from the max/minPPS by substituting 
their known values into the equations for other variables. 
This gives us a max/minPPS with LFP $0 < q'^* < 1$ and does not increase 
$|P|$. Now we can use the 
iterated GNM, with rounding down, as outlined earlier in this section.
In each iteration of GNM we solve an LP. 
Each LP has at most $n \leq |P|$ variables, at most $2n$ equations and the numerators and denominators of each rational coefficient are no larger than 
$2^{j+2+4|P|}$, so it can be solved in time polynomial in $|P|$ and $j$ using standard algorithms. We need only $j+2+4|P|$ iterations involving one LP each. 
Putting back the removed $0$ and $1$ values into the resulting vector gives us the full result $q^*$. This can all be done in polynomial time.\end{proof}

\section{Computing an $\epsilon$-optimal policy in P-time}

First let us note that we can not hope to compute an optimal policy
in P-time, without a major breakthrough:

\begin{theorem} \label{optimalhard} 
Computing an optimal policy for a max/minPPS is PosSLP-hard.\end{theorem}
\begin{proof}
Recall from \cite{rmc,ESY12} that the termination
probability vector $q^*$ of a SCFG (equivalently, of a 1-exit RMC) 
can be equivalently
viewed as the LFP of a purely probabilistic PPS, and vice-versa.

It was shown in \cite{rmc} (Theorems 5.1 and 5.3), that given a PPS
(equivalently, a SCFG or 1-RMC), and given a rational probability $p$,
it is PosSLP-hard to decide whether the LFP $q^*_1 > p$, for a given
rational $p$, as well as to decide whether $q^*_1 < p$.  (In fact,
these hardness results hold already even if $p = 1/2$.)

The fact that computing an optimal policy for max/minPPS  
is PosSLP-hard follows easily from this:
For the case of maxPPSs  (minPPS, respectively),
given a PPS, x=P(x), and given $p$, we simply add a new variable
$x_0$ to the PPS, and a corresponding equation:

\begin{equation}\label{eq:posslp-hardness}
  x_0 =  \max \{ p, x_1 \}    \quad  ( =  \min \{p, x_1 \} )
\end{equation}

It is clear that $q^*_i > p$  ($q^*_i < p$, respectively) 
for the original PPS, 
if and only if in any optimal policy $\sigma$,
for the augmented maxPPS (minPPS, respectively),
the policy picks $x_1$ rather than $p$ 
on the RHS of equation
\ref{eq:posslp-hardness}.
So, if we could compute an optimal policy for a maxPPS (minPPS), 
we would be able to decide whether $q^*_i > p$  (whether $q^*_i < p$,
respectively).
\end{proof}

Since we can not hope to compute an optimal policy for
max/minPPSs in P-time without a major breakthrough, we will instead seek
to find a policy $\sigma$ such that $\|q^*_\sigma - q^*\|_\infty \leq
\epsilon$ for a given desired $\epsilon>0$, 
in time {\em poly}($|P|$, $\log(1/\epsilon)$). 
We have an algorithm for approximating $q^*$. Can we use
a sufficiently close approximation, $q$, to $q^*$ to find such an
$\epsilon$-optimal strategy? Once we have an approximation
$q$, it seems natural to consider policies
$\sigma$ such that $P_\sigma(q) = P(q)$. For minPPSs, this means choosing
the variable that has the lowest {\em approximate} value $q_i$ and for
maxPPS choosing the variable that has the highest {\em approximate} value.  
It turns out that this works as long as we
can establish good enough upper bounds on the norm of $(I-P'_\sigma(x))^{-1}$ for certain values of $x$.   
Recall that for a square matrix $A$, $\rho(A)$ denotes its spectral radius. 
For a vector $x$, the $l_\infty$ norm is $\|x\|_\infty := \max_i |x_i|$,
and its associated matrix norm $\| A \|_\infty$  is the maximum 
absolute-value row sum of $A$, i.e., $\| A \|_\infty := \max_i \sum_j |A_{i,j}|$.

\begin{theorem} \label{neednorm} For a max/minPPS, $x=P(x)$, 
given $0 \leq q \leq q^*$, such that $q < 1$, 
and a policy $\sigma$ such that $P(q) = P_\sigma(q)$,  and such that 
$\rho(P'_\sigma(\frac{1}{2}(q^* + q^*_\sigma))) < 1$, and thus
$(I-P'_\sigma(\frac{1}{2}(q^* + q^*_\sigma)))^{-1}$ exists  and
is non-negative,
then 
\[\|q^*_\sigma - q^*\|_\infty \leq (2\|(I - P'_\sigma(\frac{1}{2}(q^*_\sigma + q^*)))^{-1}\|_\infty +1) \|q^*-q\|_\infty\]\end{theorem}
 \begin{proof} 
 We know that $q$ is close to $q^*$. We just have to show that $q$ is close to $q^*_\sigma$ as well. 
We have to exploit some results about PPSs established in \cite{ESY12}.
 \begin{lemma} \label{perfectnewton} 
If $x=P(x)$ is a PPS, 
with LFP $q^*$, such that $0 < q^* \leq 1$,
and $0 \leq y \leq q^*$, such that $y < 1$, then:
 $$q^*-y = (I - P'(\frac{1}{2}(q^* + y)))^{-1}(P(y) - y)$$
 \end{lemma}
\begin{proof}
Lemma 3.3 of \cite{ESY12} tells us that for any PPS, $x=P(x)$,  (assumed to be in SNF form), and
any pair of vectors $a,b \in \real^n$, we have $P(a) - P(b) = P'((a+b)/2) (a-b)$.
Applying this to $a=q^*$ and $b=y$,  we have that  
\[q^*-P(y) = P'((1/2)(q^*+y))(q^*-y)\]
Subtracting both sides from $q^*-y$, we have that:
\begin{equation}\label{eq:revised-equality-for-epsilon}
P(y) - y = (I-P'((1/2)(q^*+y))) (q^*-y)
\end{equation}
Now, by Lemma \ref{invnonneg}, we know that for any $z \leq q^*$, such that $z < 1$,
$(I-P'(z))^{-1}$ exists and is non-negative.
But since $y \leq q^*$,  clearly also $(1/2)(q^*+y) \leq q^*$, and since $y < 1$,
and $q^* \leq 1$, then clearly $(1/2)(q^*+y) < 1$.
Thus $(I-P'((1/2)(q^*+y))^{-1}$ exists and is non-negative.
Multiplying both sides of equation (\ref{eq:revised-equality-for-epsilon}) by
$(I-P'((1/2)(q^*+y))^{-1}$, we obtain:
\[ q^*-y= (I-P'(1/2(q^*+y))^{-1} (P(y) - y) \]
as required. 
\end{proof}

By assumption, $\sigma$ was chosen such that $P(q) = P_\sigma(q)$.
Note also that since $0 \leq q \leq q^*$, we have 
$0 \leq P'_\sigma(\frac{1}{2}(q + q^*_\sigma)) \leq
P'_\sigma(\frac{1}{2}(q^* + q^*_\sigma)$,
and thus $0 \leq \rho(P'_\sigma(\frac{1}{2}(q + q^*_\sigma))) \leq
\rho(P'_\sigma(\frac{1}{2}(q^* + q^*_\sigma)) < 1$.
Thus $(I- (P'_\sigma(\frac{1}{2}(q + q^*_\sigma)))^{-1}$ also exists
and is non-negative.
Using this, and applying Lemma 
\ref{perfectnewton} to the PPS  $x=P_\sigma(x)$, 
where we set $y :=q$,
and taking norms, we obtain the following inequality:
\begin{equation}\label{eq:inequality-on-norms-for-e-opt-policy}
\|q^*_\sigma-q\|_\infty \leq \|(I - P'_\sigma(\frac{1}{2}(q^*_\sigma + q)))^{-1}\|_\infty \|P(q) - q\|_\infty
\end{equation}

  To find a bound on  $\|P(q) - q\|_\infty$, we need the following:

  \begin{lemma} \label{lipschitz} 
If $x=P(x)$ is a max/minPPS, and if $0 \leq y \leq q^*$, 
then
$\|P(y) - y\|_\infty \leq 2 \|q^* -y\|_\infty$. \end{lemma}
   \begin{proof} Suppose that $x=P(x)$ is a PPS. 
By Lemma 3.3 of \cite{ESY12}, we have that $q^*-P(y) = P'(\frac{1}{2}(y+q^*))(q^*-y)$. Since $\frac{1}{2}(y+q^*) \leq 1$, $\|P'(\frac{1}{2}(y+q^*))\|_\infty \leq 2$: If the $i$th row has $x_i=P_i(x)$ of type L then $\sum_{j=1}^n |p_{i,j}| \leq 1$ and if $x_i=P_i(x)$ has type Q, then $\sum_{j=1}^n |\frac{\partial P_i(x)}{\partial x_j}(\frac{1}{2}(y+q^*))| = \frac{1}{2}(y_j+q^*_j) + \frac{1}{2}(y_k+q^*_k) \leq 2$.
   So we have that $\|q^*-P(y)\|_\infty \leq \|P'(\frac{1}{2}(y+q^*))\|_\infty\|q^*-y\|_\infty \leq 2\|q^*-y\|_\infty$.
  As well as $y \leq q^*$, we know that $P(y) \leq q^*$ since $P(x)$ is monotone. If $(P(y))_i \leq y_i$, then $y_i - P(y)_i \leq q^*_i - P(y)_i \leq  \|q^*-P(y)\|_\infty \leq 2\|q^*-y\|_\infty$. If $P_i(y) \geq y_i$, $P_i(y)-y_i \leq q^*_i -y_i \leq \|q^*-y\|_\infty$. So $\|P(y) - y\|_\infty \leq 2 \|q^* -y\|_\infty$ as required.
  
  If $x=P(x)$ is a max/minPPS, then it has some optimal policy, $\tau$, 
and from the above, $\|P_\tau(y) - y\|_\infty \leq 2 \|q^* -y\|_\infty$. 
It thus only remains to show that $|P_i(y) - y_i| \leq 2 \|q^* -y\|_\infty$ when $x_i = P_i(x)$ is of form M
(because the other equations don't change in $x=P_\tau(x)$). 

If $P_i(y) \geq y_i$, then this is follows easily: as before we have that 
$P_i(y)-y_i \leq q^*_i -y_i \leq \|q^*-y\|_\infty$. 
Suppose that instead we have $P_i(y) \leq y_i$.
Then we consider the two cases (min and max) separately:

Suppose $x=P(x)$ is a minPPS, 
and that $P_i(x) = \text{min }\{x_j,x_k\}$. 
Since $q^* = P(q^*)$, we have:

\begin{equation}\label{eq:min-case-diff-y-py-less-than-diff-qstar-y}
0 \leq y_i - P_i(y)  \leq q^*_i - P_i(y) 
= \min \{q^*_j,q^*_k\} -  P_i(y)
\end{equation}

We can assume, w.l.o.g., that $P_i(y) \equiv \min \{y_j,y_k\} = y_j$.
(The case where $P_i(y) = y_k$ is entirely analogous.)
Then, by (\ref{eq:min-case-diff-y-py-less-than-diff-qstar-y}), 
we have:
$$ 0 \leq y_i - P(y)_i  \leq \min \{q^*_j,q^*_k\} - y_j  \leq q^*_j - y_j
\leq \|q^*-y\|_\infty$$

Suppose now that $x=P(x)$ is a maxPPS, and
that $P_i(x) \equiv \text{max }\{x_j,x_k\}$.  
Again, we are already assuming that  $P_i(y) \leq y_i$.  
Since $q^* = P(q^*)$, we have:

\begin{equation}\label{eq:max-case-diff-y-py-less-than-diff-qstar-y}
0 \leq y_i - P_i(y)  \leq q^*_i - P_i(y) 
= P_i(q^*) -  \max \{y_j, y_k\}
\end{equation}

We can assume, w.l.o.g., that $P_i(q^*) \equiv \max\{q^*_j, q^*_k\}
= q^*_j$.
(Again, the case when $P_i(q^*) = q^*_k$ is entirely analogous.)
Then, by (\ref{eq:max-case-diff-y-py-less-than-diff-qstar-y}), 
we have:

$$0 \leq y_i - P_i(y)  \leq q^*_j - \max \{y_j,y_k\}  \leq q^*_j - y_j
\leq \|q^*-y\|_\infty$$

This completes the proof of the Lemma for all max/minPPSs.\end{proof}

   Now, we can show the result:
   \begin{eqnarray*} \|q^* - q^*_\sigma\|_\infty  & \leq & \|q^* - q\|_\infty  + \|q^*_\sigma - q\|_\infty \\
											& \leq & \|q^* - q\|_\infty  + \|(I - P'_\sigma(\frac{1}{2}(q^*_\sigma + q)))^{-1}\|_\infty \|P_\sigma(q) - q\|_\infty \\
											& = & \|q^* - q\|_\infty  + \|(I - P'_\sigma(\frac{1}{2}(q^*_\sigma + q)))^{-1}\|_\infty  \|P(q) - q\|_\infty \\
											& \leq & \|q^* - q\|_\infty  + \|(I - P'_\sigma(\frac{1}{2}(q^*_\sigma + q)))^{-1}\|_\infty   2 \|q^* - q\|_\infty \\
											& = & (2\|(I - P'_\sigma(\frac{1}{2}(q^*_\sigma + q)))^{-1}\|_\infty +1) \|q^* - q\|_\infty \\
											& \leq & (2\|(I - P'_\sigma(\frac{1}{2}(q^*_\sigma + q^*)))^{-1}\|_\infty +1) \|q^* - q\|_\infty\end{eqnarray*}
The last inequality follows because $q \leq q^*$, and
$$0 \leq (I-P'_\sigma(q^*_\sigma + q))^{-1} =  
\sum^\infty_{i=0} (P'_\sigma(q^*_\sigma + q))^i \leq 
\sum^\infty_{i=0} (P'_\sigma(q^*_\sigma + q^*))^i = 
(I-P'_\sigma(q^*_\sigma + q^*))^{-1}.$$
\end{proof}

Finding these bounds is different for maxPPSs and minPPSs . Although
we assume that $0 < q^* < 1$, for an arbitrary policy $\sigma$, it
need not be true that $0 < q^*_\sigma < 1$. But the following
obviously does hold:

\begin{proposition} \label{oneoftwo}
  Given a max/minPPS, $x=P(x)$, with LFP $q^*$ such that 
$0 < q^* < 1$,  for any policy $\sigma$: \\
  (i) If $x=P(x)$ is a maxPPS then $q^*_\sigma < 1$.\\
  (ii) If $x=P(x)$ is a minPPS, then $q^*_\sigma > 0$.
\end{proposition}
\begin{proof}
This is trivial:
if $x=P(x)$ is a maxPPS, then clearly $q^*_\sigma \leq q^* < 1$,
because $\sigma$ can be no better than an optimal strategy. 
Likewise, if $x=P(x)$ is a minPPS, then $0 < q^* \leq q^*_\sigma$,
for the same reason.
\end{proof}

For maxPPSs, we may have that some coordinate of $q^*_\sigma$ is 
equal to $0$ and
for minPPSs we may have that some coordinate of $q^*_\sigma$ is
equal to $1$, even when $0 < q^* < 1$. 
This is the source of the different complications.
We prove the following result in the appendix:

\begin{theorem} \label{normbounds}
If $x=P(x)$ is a PPS with LFP $q^* > 0$  then\\
(i) If $q^* < 1$ and $0 \leq y < 1$, then 
 $(I-P'(\frac{1}{2}(y+q^*)))^{-1}$ exists and is non-negative, and
$$\|(I-P'(\frac{1}{2}(y+q^*)))^{-1}\|_\infty \leq  2^{10|P|} \text{max } \{2(1-y)_{\min}^{-1}, 2^{|P|}\}$$
(ii) If $q^* = 1$ and $x = P(x)$ is strongly connected 
(i.e. every variable depends directly or indirectly  on every other) 
and $0 \leq y <  1 = q^*$, then
$(I -P'(y))^{-1}$ exists and is non-negative, and 
$$\|(I -P'(y))^{-1}\|_\infty \leq 2^{4|P|} \frac{1}{(1-y)_{\min}}$$
\end{theorem}

We first focus on minPPSs, 
for which we shall show that if $y$ is a close approximation to $q^*$, 
then any policy $\sigma$ with $P(y) = P_\sigma(y)$ is $\epsilon$-optimal.
The maxPPS case will not be so simple: the analogous statement is false
for maxPPSs.

\begin{theorem} \label{mineopt} If $x=P(x)$ is a minPPS, 
with LFP $0 < q^* < 1$, and $0 \leq \epsilon \leq 1$, and  $0 \leq y \leq q^*$, 
such that $\|q^* - y\|_\infty \leq 2^{-14|P| - 3} \epsilon$, then for any policy $\sigma$ with $P_\sigma(y) = P(y)$,  
$\|q^* - q^*_\sigma\|_\infty \leq \epsilon$.
\end{theorem}
\begin{proof}
By Proposition \ref{oneoftwo},
 $q^*_\sigma \geq q^*$, and so $q^*_\sigma > 0$.  
Suppose for now that $q^*_\sigma < 1$ (we will show this later). 
Then applying Theorem \ref{normbounds} (i), 
for the case where we set $y:=q^*$ and the PPS is $x = P_\sigma(x)$, yields that
$$\|(I-P'_\sigma(\frac{1}{2}(q^*+q^*_\sigma)))^{-1}\|_\infty \leq 2^{10|P_\sigma|}\text{max }\{\frac{2}{(1-q^*)_{\min}},2^{|P|}\}$$
Note that $|P_\sigma| \leq |P|$. 
Since for any minPPS, $x=P(x)$, there is an optimal strategy $\tau$, and
$x=P_\tau(x)$ is a PPS with the same LFP, $q^*_\tau = q^*$, as $x=P(x)$, and 
furthermore since $|P_\tau| \leq |P|$,  
it follows from Theorem 3.12 of \cite{ESY12} that
$(1-q^*)_{\min} \geq 2^{-4|P|}$. Thus
$$\|(I-P'_\sigma(\frac{1}{2}(q^*+q^*_\sigma)))^{-1}\|_\infty \leq 2^{14|P|+1}$$
Theorem \ref{neednorm} now gives that
$$\|q^* - q^*_\sigma\|_\infty \leq (2^{14|P| + 2}+1) \|q^* -y\|_\infty \leq \epsilon$$
Thus, under the assumption that $q^*_\sigma <1$, we are done.

To complete the proof, 
we now show that $q^*_\sigma < 1$. 
Suppose, for a contradiction, that for some $i$, $(q^*_\sigma)_i = 1$. Then by results in \cite{rmc}, $x=P_\sigma(x)$ has a bottom strongly connected component $S$ with $q^*_S = 1$. If $x_i$ is in $S$ then only variables in $S$ appear in $(P_\sigma)_i(x)$, so we write $x_S=P_S(x)$ for the PPS which is formed by such equations. We also have that $P'_S(1)$ is irreducible and that the least fixed point solution of $x_S=P_S(x_S)$ is $q^*_S = 1$.
Take $y_S$ to be the subvector of $y$ with coordinates in $S$. Now 
if we apply Theorem \ref{normbounds} (ii), by taking the $y$ in its 
statement to be $\frac{1}{2}(y_S + 1)$, it 
gives that
$$\|(I -P'_S(\frac{1}{2}(y_S + 1)))^{-1}\|_\infty \leq 2^{4|P_S|} \frac{1}{\frac{1}{2}(1-y_S)_{\min}}$$
But $|P_S| \leq |P|$ and $(1-y_S)_{\min} \geq (1-q^*)_{\min} \geq 2^{-4|P|}$. Thus
$$\|(I -P'_S(\frac{1}{2}(y_S + 1)))^{-1}\|_\infty \leq 2^{8|P| + 1}$$
Lemma \ref{perfectnewton} gives that
 $$1-y_S = (I - P'_S(\frac{1}{2}(1 + y_S)))^{-1}(P_S(y_S) - y_S)$$
 Taking norms and re-arranging gives:
 $$\|P_S(y_S) - y_S)\|_\infty \geq \frac{\|1-y_S\|_\infty}{\|(I -P'_S(\frac{1}{2}(y_S + 1)))^{-1}\|_\infty} \geq \frac{2^{-4|P|}}{2^{8|P| + 1}} \geq 2^{-12|P|-1}$$
 However $\|P_S(y_S) - y_S)\|_\infty \leq \|P_\sigma(y) -y\|_\infty$ and $P_\sigma(y) =P(y)$. We deduce that $\|P(y)-y\|_\infty \geq 2^{-12|P|-1}$. 
Lemma \ref{lipschitz} states that $\|P(y) - y\|_\infty \leq 2 \|q^* -y\|_\infty$. 
We thus have $\|q^* -y\|_\infty \geq 2^{-12|P|-2}$. This contradicts our assumption that 
$\|q^* -y\|_\infty \leq 2^{-14|P| - 3} \epsilon$ for some $\epsilon \leq 1$. 
 \end{proof}

Now we proceed to the harder case of maxPPSs. The main theorem in this case
is the following.

 \begin{theorem} \label{maxeopt} If $x=P(x)$ is a maxPPS with $0 < q^* < 1$ and given 
$0 \leq \epsilon \leq 1$ and a vector $y$, with $0 \leq y \leq q^*$, such that 
$\|q^* - y \|_\infty \leq 2^{-14|P| - 2} \epsilon$, there exists a policy $\sigma$ such that  
$\|q^* - q^*_\sigma\|_\infty \leq \epsilon$, and furthermore, 
such a policy can be computed in P-time, given $x=P(x)$ and $y$.
\end{theorem}
We need a policy $\sigma$ for which we can apply Theorem \ref{normbounds}, and 
for which we can get good bounds on $\|P_\sigma(y) - y\|_\infty$. 
Firstly we show that such policies exist. 
In fact, any optimal policy will do: for an optimal policy $\tau$, $q^*_\tau > 0$ and 
Lemma \ref{lipschitz} applied to $x = P_\tau(x)$ gives that $\|P_\tau(y) - y\|_\infty \leq 2^{-14|P|-1} \epsilon$.
 Unfortunately the optimal policy might be hard to find (Theorem \ref{optimalhard}). We can however, 
given a policy $\sigma$ and the PPS $x=P_\sigma(x)$, easily detect in polynomial time 
whether $q^*_\sigma > 0$ (see, e.g., Theorem 2.2 of \cite{rmc}, and also \cite{rsm}).
We shall also make use of the following easy fact:

\begin{lemma} \label{detect0} If $x=P(x)$ is a PPS with n variables, and with LFP $q^*$, 
then for any variable index $i \in \{1,\ldots,n\}$ the following are equivalent\\
(i) $q^*_i > 0$.\\
(ii) there is a $k >0$ such that $(P^k(0))_i > 0$.\\
(iii) $(P^n(0))_i > 0$.\end{lemma}
\begin{proof} (i) $\implies$ (ii): From \cite{rmc}, $P^k(0) \rightarrow q^*$ as $k \rightarrow \infty$. 
It follows that if $(P^k(0))_i = 0$ for all $k$, then $q^*_i = 0$.\\
(ii)  $\implies$ (iii): Firstly, if there is a $1 \leq k < n$ with $(P^k(0))_i > 0$ then $(P^n(0))_i > 0$. 
$P(0) \geq 0$ and so by monotonicity and an easy induction $P^{l+1}(0) \geq P^{l}(0)$ for all $l > 0$. 
Another induction gives that $P^m(0) \geq P^l(0)$ when $m \geq l > 0$. As $k < n$, $(P^n(0))_i \geq (P^k(0))_i > 0$.

Whether $P_i(x) > 0$ depends only on whether each $x_j >0$ or not and not on the value of $x_j$. 
So, for any $k$, whether $(P^{k+1}(0))_i > 0$ depends only on the set $S_k = \{x_j \text{ such that } (P^{k}(0))_j > 0\}$. 
From before $P^{k+1}(0) \geq P^k(0)$, so $S_{k+1} \supseteq S_k$. If ever we have that $S_{k+1} = S_k$, 
then for any $j$, $(P^{k+2}(0))_j > 0$ whenever $(P^{k+1}(0))_j > 0$ 
so $S_{k+2} = S_{k+1} = S_k$. $S_{k+1} \supset S_k$ can only occur for $n$ values of $k$ as there are only $n$ variables to add. 
Consequently $S_{n+1} = S_n$ and so $S_m = S_n$ whenever $m > n$. 
So if we have a $k > n$ with $(P^k(0))_i > 0$, then $(P^n(0))_i > 0$

(iii) $\implies$ (i): By monotonicity and an easy induction, $q^* \geq P^k(0)$ for all $k >0$. 
In particular $q^* \geq P^n(0)$.  So $q^*_i \geq (P^n(0))_i > 0$.\end{proof}

Given the maxPPS, $x=P(x)$, with $0 < q^* < 1$, and given a vector $y$
that satisfies the conditions of Theorem \ref{maxeopt}, 
we shall use the following algorithm to obtain the policy we need:

\begin{enumerate}
\item Initialize the policy $\sigma$ to any policy such that $P_\sigma(y) = P(y)$.

\item Calculate for which variables $x_i$ in $x=P_\sigma(x)$ we have $(q^*_\sigma)_i = 0$. 
Let $S_{0}$ denote this set of variables.
(We can do this in P-time;  see e.g., Theorem 2.2 of \cite{rmc}.) 

\item If for all $i$ we have $(q^*_\sigma)_i > 0$, i.e., if $S_{0} = \emptyset$, then terminate
and output the policy $\sigma$.

\item Otherwise, look for a variable $x_i$, where $P_i(x)$ is of form M, 
with $P_i(x) = \text{max }  \{x_j,x_k\}$, and where 
$(q^*_\sigma)_i = 0$ but one of $x_j , x_k$, say $x_j$, has $(q^*_\sigma)_j > 0$ and where
furthermore $\|y_i - y_j\| \leq 2^{-14|P|-1} \epsilon$. (We shall establish that such a 
pair $x_i$ and $x_j$ will always exist when we are at this step of the algorithm.)
 
Let $\sigma'$ be the policy that chooses $x_j$ at $x_i$ but is otherwise identical to $\sigma$. 
Set $\sigma := \sigma'$ and return to step 2.
\end{enumerate}

\begin{lemma} The steps of the above algorithm are always well-defined, and
the algorithm always terminates with a policy $\sigma$ 
such that $q^*_\sigma > 0$ and $\|P_\sigma(y) - y\|_\infty \leq 2^{-14|P|-1} \epsilon$.
\end{lemma}
\begin{proof}
Firstly, to show that the steps of the algorithm are always well-defined,
we need to show that if there exists an $x_i$ with $(q^*_\sigma)_i = 0$, then step 4 will 
find some variable to switch to. 
Suppose there is such an $x_i$. Let $\tau$ be an optimal policy. $(q^*_\tau)_i =  q^*_i >  0$. 
So by Lemma \ref{detect0}, $(P_\tau^n)_i > 0$. 
For any variable $x_j$ with $(P_\tau(0))_j > 0$, the equation $x_j = P_j(x)$ must have form L and not M 
so $(P_\sigma(0))_j > 0$ and so $(q^*_\sigma)_j > 0$.
There must be a least $k$, $k_{\min}$ with $1 < k_{\min} \leq n$, such that there is a variable $x_j$ 
with $(P_\tau^k(0))_j > 0$ but $(q^*_\sigma)_j = 0$. 
Let $x_{i'}$ be a variable such that $(P_\tau^{k_{\min}}(0))_{i'} > 0$ 
but $(q^*_\sigma)_{i'} = 0$.

Suppose that $x_{i'} = P_{i'}(x)$ has form Q, then $P_{i'}(x) =
x_jx_l$ for some variables $x_j$, $x_l$. 
We have $0 < (P_{\tau}^{k_{\min}}(0))_{i'} =
(P_{\tau}^{k_{\min}-1}(0))_j(P_{\tau}^{k_{\min}-1}(0))_l$. So
$(P_{\tau}^{k_{\min}-1}(0))_j > 0$ and $(P_{\tau}^{k_{\min}-1}(0))_l > 0$. The
minimality of $k_{\min}$ now gives us that $(q^*_\sigma)_j > 0$ and
$(q^*_\sigma)_l > 0$. So $ (q^*_\sigma)_{i'} = (q^*_\sigma)_j
(q^*_\sigma)_l > 0$. This is a contradiction. Thus, $x_{i'} = P_{i'}(x)$
does not have form Q.

Similarly, $x_{i'} = P_{i'}(x)$ does not have form L. So $x_{i'} =
P_{i'}(x)$ has form M. There are variables $x_j$, $x_l$ with
$P_{i'}(x) = \text{max } \{x_j,x_l\}$. Suppose, w.l.o.g. that
$(P_\tau(x))_{i'} = x_j$. 
We have $P^{k_{\min}}_\tau(0))_{i'} > 0$ and so
$(P^{k_{\min} -1}(0))_j > 0$. By minimality of $k_{\min}$, we have that
$(q^*_\sigma)_j > 0$. We have that $(q^*_\sigma)_{i'} = 0$ and so
$(P_\sigma(x))_{i'} = x_l$.

Lemma \ref{lipschitz} applied to the system $x = P_\tau(x)$ gives that 
$\|P_\tau(y) - y\|_\infty \leq 2^{-14|P|-1} \epsilon$. 
So $|y_{i'} - y_{j}| = |y_{i'} - (P_\tau(y))_{i'}| \leq 2^{-14|P|-1} \epsilon$. 
Thus, step 4 could use $i'$ 
and change the policy $\sigma$ at $i'$ (i.e., switch $\sigma(i')$) from $x_l$ to $x_j$.

\noindent Next, we need to show that the algorithm terminates:

\begin{claim} If step 4 switches the variable $x_i$ with $P_i(x) = \text{max } \{x_j,x_k\}$ 
from $(P_\sigma(x))_i = x_k$ to $(P_{\sigma'}(x))_i = x_j$, then \\ 
(i) $q^*_{\sigma'} \geq q^*_\sigma$,\\
(ii) ($q^*_{\sigma'})_i > 0$,\\
(iii) The set of variables  $x_l$ with $(q^*_{\sigma'})_l > 0$ is a strict superset of the set of variables 
$x_l$ with $(q^*_\sigma)_l > 0$. \end{claim}
\begin{proof}
Recall that step 4 will only switch if $(q^*_\sigma)_i = 0$ and $(q^*_\sigma)_j > 0$.
\begin{itemize}
\item[(i)] We show that, for any $t > 0$, $P^t_{\sigma'}(0) \geq P^t_\sigma(0)$.\\
The base case $t=1$, is clear, because the only indices $i$ where $P_i(0) \neq 0$ are
when $P_i(0)$ has form L, in which case $P_i(0) = (P_{\sigma'}(0))_i = (P_{\sigma}(0))_i$.
  
For the inductive case:
note firstly that 
$P_\sigma(x)$ and $P_{\sigma'}(x)$ only differ on the $i$th coordinate. 
$(q^*_\sigma)_i = 0$, so for any $t$, $(P^t_\sigma(0))_i = 0$.  Suppose that $P_{\sigma'}^t(0) \geq P_{\sigma}^t(0)$.  
Then by monotonicity $P_{\sigma'}^{t+1}(0) \geq P_{\sigma'}(P_{\sigma}^t(0))$.
But $(P_{\sigma'}(P_{\sigma}^t(0)))_ r = (P_{\sigma}^{t+1}(0))_ r$ when $r \not= i$. 
Furthermore, $(P_{\sigma'}(P_{\sigma}^t(0)))_i \geq 0 = (P_{\sigma}^{t+1}(0))_ i$.
 So  $P_{\sigma'}(P_{\sigma}^k(0)) \geq P_{\sigma}^{k+1}(0)$. 
We thus have that $P_{\sigma'}^{t+1}(0) \geq P_{\sigma}^{t+1}(0)$. 

We know that as $t \rightarrow \infty$, $P^t_{\sigma'}(0) \rightarrow q^*_{\sigma'}$ and $P^t_\sigma(0) \rightarrow q^*_\sigma$. So $q^*_{\sigma'} \geq q^*_\sigma$.

\item[(ii)] We have $(q^*_{\sigma'})_i  = (q^*_{\sigma'})_j$. By (i) $(q^*_{\sigma'})_j \geq (q^*_{\sigma})_j$. 
We chose $x_j$ such that $(q^*_\sigma)_j > 0$. So $(q^*_{\sigma'})_i > 0$.

\item[(iii)] If $(q^*_\sigma)_l > 0$, then by (i) $(q^*_{\sigma'})_l > 0$. Also $(q^*_\sigma)_i = 0$ and by (ii) $(q^*_{\sigma'})_i > 0$.
\end{itemize}\end{proof}

Thus, if at some stage of the algorithm we do not yet have $q^*_\sigma > 0$, then step 4 always 
gives us a new $\sigma'$ with more coordinates having $(q^*_{\sigma'})_i > 0$. 
Furthermore, note that 
if $\|P_\sigma(y) - y\|_\infty \leq 2^{-14|P|-1} \epsilon$ then $\|P_{\sigma'}(y) - y\|_\infty \leq 2^{-14|P|-1} \epsilon$.
Our starting policy has $\|P_\sigma(y) - y\|_\infty = \|P(y) - y\|_\infty\leq 2^{-14|P|-1} \epsilon$.
The algorithm terminates and gives a $\sigma$ with $q^*_\sigma > 0$ and $\|P_\sigma(y) - y\|_\infty \leq 2^{-14|P|-1} \epsilon$.
\end{proof}

We can now complete the proof of the Theorem:

\begin{proof}[Proof of Theorem \ref{maxeopt}.]
Using the algorithm, we find a $\sigma$ with $\|y - P_\sigma(y)\|_\infty \leq 2^{-14|P|-1} \epsilon$ and $q^*_\sigma > 0$.
By Proposition \ref{oneoftwo}, $q^*_\sigma < 1$.  
Applying Theorem \ref{normbounds} (i) to the PPS $x=P_\sigma(x)$  and point $y:=q^*$
(not to be confused with the $y$ in the statement of Theorem \ref{maxeopt}), gives that
$$\|(I-P'_\sigma(\frac{1}{2}(q^*+q^*_\sigma)))^{-1}\|_\infty \leq 2^{10|P_\sigma|}\text{max } \{\frac{2}{(1-q^*)_{\min}}, 2^{|P|} \}$$
We have $|P_\sigma| \leq |P|$.
Also, from
the fact there always exists an optimal policy, 
and from Theorem 3.12 of \cite{ESY12}, 
it follows that we have $(1-q^*)_{\min} \geq 2^{-4|P|}$. So
\begin{equation}
\label{eq:0-half-final-alg-max-epsilon-case}
\|(I-P'_\sigma(\frac{1}{2}(q^*+q^*_\sigma)))^{-1}\|_\infty \leq 2^{14|P| +1}
\end{equation} 
We can not use Theorem \ref{neednorm} as stated because we need not have $P(y) = P_\sigma(y)$. 
We do however have 
\begin{equation}
\label{eq:one-half-final-alg-max-epsilon-case}
\|P_\sigma(y) - y\|_\infty \leq 2^{-14|P|-1} \epsilon
\end{equation} 
Applying Lemma \ref{perfectnewton}, and taking norms, 
we get the inequality 
\begin{equation}
\label{eq:second-half-final-alg-max-epsilon-case}
\|q^*_\sigma-y\|_\infty \leq \|(I - P'(\frac{1}{2}(q^*_\sigma + y)))^{-1}\|_\infty \|P(y) - y\|_\infty
\end{equation}
Combining (\ref{eq:0-half-final-alg-max-epsilon-case}), (\ref{eq:one-half-final-alg-max-epsilon-case}) 
and (\ref{eq:second-half-final-alg-max-epsilon-case}) yields:
$$\|q^*_\sigma - y\|_\infty \leq \frac{1}{2}\epsilon$$
so
$\|q^*_\sigma - q^*\|_\infty \leq \|q^*_\sigma - y\|_\infty + \|q^* - y\|_\infty \leq \frac{1}{2}\epsilon + 2^{-14|P| - 2} \epsilon \leq \epsilon$.
\end{proof}

\begin{theorem} Given a max/minPPS,  $x=P(x)$, and given $\epsilon > 0$,
we can compute an $\epsilon$-optimal policy for $x=P(x)$ in time 
poly($|P|$, $\text{log }(1/\epsilon)$)\end{theorem}           

\begin{proof} First we use the algorithms from \cite{rmdp} to detect variables $x_i$ with $q^*_i =0$ or $q^*_i = 1$ in time polynomial in $|P|$. Then we can remove these from the max/minPPS by substituting the known values into the equations for other variables. This gives us an max/minPPS with least fixed point $0 < q'^* < 1$ and does not increase $|P|$. 
To use either Theorem \ref{maxeopt} or Theorem \ref{mineopt}, 
it suffices to have a $y$ with $y < q^*$ with $q^* - y \leq 2^{-14|P| - 3} \epsilon$. 
Theorem \ref{mptime} says that we can find such a $y$ in time polynomial in $|P|$ 
and $14|P| - \text{log }(\epsilon)$, which is polynomial in $|P|$ and $\text{log }(1/\epsilon)$ as required. 
Now depending on whether we have a maxPPS or minPPS,
Theorem \ref{maxeopt} or Theorem \ref{mineopt}  show that from this $y$, we can find an 
$\epsilon$-optimal policy for the max/minPPS with $0 < q'^* < 1$ in time polynomial in $|P|$ and 
$\text{log }(1/\epsilon)$. All that is left to show is that we can extend this policy to the variables 
$x_i$ where $q^*_i =0$ or $q^*_i=1$ while still remaining $\epsilon$-optimal.

We next show how this can be done.

For a minPPS, if $q^*_i=1$ then for any policy
$\sigma$, $(q^*_\sigma)_i = 1$ so the choice made at such variables $x_i$ is
irrelevant. Similarly, for maxPPSs, when $q^*_i =0$, any choice at $x_i$ is
optimal.

For a minPPS with $q^*_i=0$,  if $P_i(x)$ has form M,
we can choose any variable $x_j$ with $q^*_j = 0$. There is such a variable: 
if $P_i(x) = \text{min }\{x_j,x_k\}$ and $q^*_i=0$ then either $q^*_j=0$ or $q^*_k=0$. 
Let $\sigma$ be a policy such that for each variable $x_i$ with $q^*_i=0$, $(q^*)_{\sigma(i)} = 0$. 
We need to show that $(q^*_\sigma)_i=0$ for all such variables. Suppose that, for some $k \geq 0$, 
$(P^k_\sigma(0))_i = 0$ for all $x_i$ such that $q^*_i=0$. Then $P(P^k_\sigma(0))_i = 0$
for all $x_i$ with $q^*_i=0$. 

To see why this is so, note that whether or not $P_i(z) =0$ depends only on which coordinates of $z$ are $0$,
and furthermore if $P_i(z)=0$ when the set of $0$ coordinates of $z$ is $S$,
then for any vector $z'$ where the $0$ coordinates of $z'$ are $S' \supseteq S$, 
we have $P_i(z') = 0$. 
Since the coordinate $S$ that are $0$ in $q^*$ 
are a subset of the coordinates $S'$ that are $0$ in $P^k_\sigma(0)$, 
and we have $P_i(q^*)= q^*_i=0$, we thus have 
$P(P^k_\sigma(0))_i = 0$.

If $P_i(x) = \text{min }\{x_j,x_k\}$ and $q^*_i=0$ then either $q^*_j = 0$ or $q^*_k= 0$. 
Suppose w.l.o.g. that $(P_\sigma(x))_i =x_j$.  Then $q^*_j = 0$, so by assumption $(P^k_\sigma(0))_j = 0$ and so $(P_\sigma(P^k_\sigma(0)))_i = 0$. We now have enough for $(P^{k+1}_\sigma(0))_i = 0$ 
for each variable $x_i$ with $q^*_i=0$. $P^0_\sigma(0) = 0$, so by induction for all $k \geq 0$, $(P^k_\sigma(0))_i = 0$ 
for all $x_i$ with $q^*_i=0$. From this, for each variable $x_i$ with $q^*_i=0$, $(q^*_\sigma)_i=0$.
  
The case of a maxPPS that have variables with $q^*_i=1$ is not so
simple. The P-time algorithm given in \cite{rmdp} to detect vertices
with $q^*_i=1$, produces a partial randomised policy for such vertices
(Lemma 12 in \cite{rmdp}).  A randomised policy is a map $\rho: M
\rightarrow [0,1]$, that turns a maxPPS $x=P(x)$ into a PPS
$x=P_\rho(x)$ by replacing equations of form M, $P_i(x)=\text{max
}\{x_j,x_k\}$, with equations of form L $P_i(x)=\rho(i)x_j +
(1-\rho(i))x_k$. We would prefer a non-randomised (pure) policy
$\sigma$ with $(q^*_\sigma)_i =1$ for all variables $x_i$ with $q^*_i =
1$. Theorem \ref{optexist} (which quotes Theorem 2 of \cite{rmdp}) guarantees the existence
of such a $\sigma$.

We can construct such a pure optimal partial policy. 
We start with $P_{(0)}(x) =P(x)$. Given an $x_i$ with $(P_{(l)}(x))_i = \text{max } \{x_j, x_k\}$ and $(q^*_{(l)})_i=1$, 
we try setting $(P_{(l+1)}(x))_i = x_j$ and see if this gives $(q^*_{(l+1)})_i = 1$. 
If it does then set $(P_{(l+1)}(x)))_i = x_j$. 
If it does not then set $(P_{(l+1)}(x)))_i = x_k$. 
We can argue inductively that the LFP $q^*_{(l)}$ of $x=P_{(l)}(x)$ is equal to
the LFP $q^*$ of $x=P(x)$ for all $l$.
The basis, $l=0$, is clear. 
For the induction step. we know from
Theorem \ref{optexist} that there is an optimal policy $\sigma$ for the maxPPS $x=P_{(l)}(x)$. 
If $\sigma$ does not have $\sigma(i) = j$ then $\sigma(i) = k$. So if setting  $(P_{(l+1)}(x))_i = x_j$ 
would not give $(q^*_{(l+1)})_i = 1$ then 
$(P_{(l+1)}(x))_i = x_k$ does give $(q^*_{(l+1)})_i = 1$. 
We have that $(q^*_{(l+1)})_r = (q^*_{(l)})_r$ when $r \not= i$ so $q^*_{(l+1)}=q^*_{(l)}$. 
When there are no $x_i$ with $(P_{(l)}(x))_i = \text{max } \{x_j, x_k\}$ and $(q^*_{(l)})_i=1$,
 we have found a pure partial optimal policy for $x_i$ with $q^*_i = 1$. 
This requires no more than $n$ calls to the polynomial time algorithm given in \cite{rmdp} 
for determining for a maxPPS, $x=P(x)$ those
coordinates $i$ such that $q^*_i =1$.
\end{proof} 

\section{Approximating the value of BSSGs in FNP}

In this section we briefly note that, as an easy corollary of
our results for BMDPs, we can obtain a TFNP
(total NP search problem) upper bound for 
computing (approximately), the {\em value} of 
{\em Branching simple stochastic games} (BSSG), where the objective
of the two players is to maximize, and minimize, the
extinction probability.
For relevant definitions and background results about 
these games see \cite{rmdp}.
It suffices for our purposes here to point out that, as shown in \cite{rmdp}, 
the value of these games (which are determined)
is characterized by the LFP solution of associated min-maxPPSs, $x=P(x)$,
where both min and max operators can occur in the equations for
different  variables.  Furthermore, 
both players have optimal policies (i.e. optimal 
pure, memoryless strategies)  
in these games (see \cite{rmdp}).

\begin{corollary}
Given a max-minPPS, $x=P(x)$, and given a rational $\epsilon > 0$,
the problem of approximating the LFP $q^*$ of $x=P(x)$,
i.e., computing a vector $v$ such that $\| q^* - v\|_\infty \leq \epsilon$,
is in TFNP, as is the problem of computing $\epsilon$-optimal
policies for both players.    
(And thus also, the problem of approximating the value,
and computing $\epsilon$-optimal strategies, for
BSSGs is in FNP.)
\end{corollary}
\begin{proof}
Given $x=P(x)$, whose LFP, $q^*$, we wish to compute,
first guess pure policies $\sigma$ and $\tau$
for the max and min players, respectively.
Then, fix $\sigma$ as max's strategy, 
and for the resulting minPPS (with LFP $q^*_\sigma$)
use our algorithm to compute in P-time an approximate value 
vector $v_\sigma$, such that $\|v_\sigma - q^*_\sigma \|_\infty \leq \epsilon/4$.
Next, fix $\tau$ as min's strategy,
and for the resulting maxPPS (with LFP $q^*_\tau$),
use our algorithm to compute in P-time an approximate value
vector $v_\tau$, such that $\| v_\tau -q^*_\tau \|_\infty \leq \epsilon/4$.
Finally,  check whether 
$\| v_{\sigma} - v_{\tau} \|_\infty \leq \epsilon/4$. 
If not, then reject this ``guess''.  If so, then 
output $\sigma$  and $\tau$ as $\epsilon$-optimal policies
for max and min, respectively,
and output $v:=v_\sigma$ (or $v :=v_\tau$) as an 
$\epsilon$-approximation of the LFP, $q^*$.   This procedure is correct because
if $q^*$ is the LFP of the min-maxPPS, $x=P(x)$, then 
$q^*_\sigma \leq q^* \leq q^*_\tau$, and thus:
\begin{eqnarray*}
   \| q^*  - v_\sigma \|_\infty & \leq &  
\| q^* - q^*_\sigma \|_\infty + \| q^*_\sigma - v_\sigma \|_\infty\\
& \leq & \| q^*_\tau - q^*_\sigma \|_\infty  
+ \| q^*_\sigma - v_\sigma \|_\infty\\
& \leq &  \| q^*_\tau - v_\tau \|_\infty + \| v_\tau - v_\sigma \|_\infty
+ \|  v_\sigma - q^*_\sigma \|_\infty + 
\| q^*_\sigma - v_\sigma \|_\infty\\
& \leq & \epsilon
\end{eqnarray*}
And likewise for $v_\tau$.
\end{proof}

It is worth noting that the problem of approximating the value of a BSSG game,
to within a desired $\epsilon > 0$, when $\epsilon$ is given as part of the input, is
already at least as hard as computing the {\em exact} value of Condon's finite-state
simple stochastic games (SSGs) \cite{condon92}, 
and thus one can not hope for a P-time
upper bound without a breakthrough.  In fact, it was shown
in \cite{rmdp} that even the {\em qualitative} problem of deciding whether
the value $q^*_i =1$ for a given BSSG (or max-minPPS), which was shown
there to be in NP$\cap$coNP,
is already at least as hard
as Condon's {\em quantitative} decision problem for finite-state
simple stochastic games.  (Whereas for finite-state SSGs the qualitative 
problem of deciding whether the value is $1$ is in P-time.)

\appendix

\section{Omitted material from Section 2}

\subsection{Proof of Lemma \ref{invnonneg}}

As usual, we always assume, w.l.o.g., that any MPS or PPS is in SNF form.
Recall that for a square matrix $A$, $\rho(A)$ denotes its spectral
radius.\\

\noindent {\bf Lemma \ref{invnonneg}.}\;
{\em 
Given a PPS, $x= P(x)$, 
with LFP $q^* >0$,  if  $0 \leq y \leq q^*$, and $y < 1$, 
then $\rho(P'(y)) < 1$, and $(I - P'(y))^{-1}$ exists and is non-negative.}\\

We first recall several closely related results 
established in our previous papers.
Recall that a PPS, $x=P(x)$, is called {\em strongly connected},
if its variable dependency graph $H$ 
is strongly connected.

\begin{lemma}{(Lemma 6.5 of \cite{rmc})}\footnote{Lemma 6.5 of  \cite{rmc} is  
actually a more general result, relating to strongly connected MPSs that
arise from more general RMCs.}\label{lem:spec-below-q-scc}
Let $x=P(x)$ be a strongly connected PPS, in $n$ variables,
with LFP $q^* > 0$.  For any vector $0 \leq y < q^*$,   
$\rho(P'(y)) < 1$, and thus $(I-P'(y))^{-1}$ exists and is nonnegative.
\end{lemma}

\begin{theorem}{(Theorem 3.6 of \cite{ESY12})} \label{thm:spec-full} 
For any PPS, $x=P(x)$,
in SNF form, which has LFP $0 < q^* < 1$, 
for all $0 \leq y \leq q^*$,
$\rho(P'(y)) < 1$ and $(I-P'(y))^{-1}$ exists and is nonnegative.
\end{theorem}

\begin{proof}[Proof of Lemma  \ref{invnonneg}]
Consider a PPS, $x=P(x)$, with LFP  $q^* > 0$, and a vector 
$0 \leq y \leq q^*$, such that $y < 1$.
Note that all we need to establish is that $\rho(P'(y)) < 1$,
because it then follows by standard facts  
(see, e.g., \cite{HornJohnson85}) that 
$(I-P'(y))^{-1}$ exists and is equal to $\sum^\infty_{i=0} (P'(y))^i \geq 0$.

Let us first show that if $x=P(x)$ is strongly connected, then
$\rho(P'(y)) < 1$.
To see this, note that if $x=P(x)$ is strongly connected, then every
variable depends on every other, and thus if 
there exists any $i \in \{1,\ldots,n\}$ such that $q^*_i < 1$, then 
it must be the case that for all $j \in \{1,\ldots,n\}$, we have 
$q^*_j < 1$. 
Thus, either $q^* = 1$, or else $0< q^* < 1$.
If $q^* = 1$, then since $y < 1$, we have $y < q^*$, and thus, 
by Lemma \ref{lem:spec-below-q-scc}, we have $\rho(P'(y)) < 1$.
If, on the other hand, $0 < q^* < 1$, then
since $0 \leq y \leq q^*$, by Theorem \ref{thm:spec-full},
we have $\rho(P'(y)) < 1$.

Next, consider an arbitrary PPS, $x=P(x)$, that 
is not necessarily strongly connected.
Recall the variable dependency graph $H$ of $x=P(x)$.
We can partition the variables into sets
$S_1,\ldots,S_k$ which form the SCCs  of $H$.
Consider the DAG, $D$, of SCCs,  whose nodes are the sets $S_i$,
and for which there is an edge from $S_i$ to $S_j$ iff
in the dependency graph $H$ 
there is a node $i' \in S_i$ with an edge to a node in $j' \in S_j$.

Consider the matrix $P'(y)$.  Our aim is to show that
$\rho(P'(y)) < 1$.   
Since we assume $q^* > 0$, $0 \leq y \leq q^*$, and $y < 1$, 
it clearly suffices to show that $\rho(P'(y)) < 1$ holds
in the case where we additionally insist that $y > 0$, 
because then for any other $z$ such that $0 \leq z \leq y$, we would have
$\rho(P'(z)) \leq \rho(P'(y)) < 1$.

So, assuming also that $y > 0$, consider the $n \times n$-matrix $P'(y)$.
To keep notation clean, we let $A := P'(y))$.  
For the $n \times n$ matrix $A$, we can consider its underlying 
{\em dependency} graph,
$H = (\{1,\ldots,n\},E_H)$, whose nodes are $\{1,\ldots,n\}$, and 
where there is an edge from $i$ to $j$ iff $A_{i,j} >0$.
Notice however that, since $ y > 0$,  this graph is precisely the same  
graph as the dependency graph $H$ of $x=P(x)$, 
and thus it has the same SCCs, and the same DAG of SCCs, $D$.
Let us sort the SCCs, 
so that we can assume $S_1,\ldots,S_k$ are topologically
sorted with respect to the partial ordering defined by the DAG $D$.
In other words, for any variable indices $i \in S_a$ and $j' \in S_b$ 
if  $(i,j) \in E_H$, then $a \leq b$.

Let $S \subseteq \{1,\ldots,n\}$ be any non-empty subset of indices, 
and let $A[S] $ denote the principle submatrix of $A$ defined by indices in $S$.
It is a well known fact that $0 \leq \rho(A[S]) \leq \rho(A)$.
(See, e.g, Corollary 8.1.20 of \cite{HornJohnson85}.)

Since $A \geq 0$,   $\rho(A)$ is an eigenvalue of $A$,
and has an associated non-negative eigenvector 
$v \geq 0$, $v \neq 0$  (again see, e.g., Chapter 8 of \cite{HornJohnson85}).
In other words, 
$$  Av = \rho(A) v $$
Firstly, if $\rho(A) = 0$, then we are of course trivially done.  So we can
assume w.l.o.g. that $\rho(A) > 0$.
Now, 
if $v_i > 0$, then for every $j$ such that $(j,i) \in E_H$,
we have $(Av)_j > 0$, and thus since $(Av)_j = \rho(A)v_j$,
we have $v_j > 0$.   Hence, repeating this argument, if $v_i > 0$ 
then for every $j$ that has a path to $i$ 
in the dependency graph $H$, we have $v_j >0$.

Since $v \neq 0$, it must be the case that there
is exists some SCC, $S_c$, of $H$ such that for every 
variable index $i \in S_c$,  $v_i > 0$,
and furthermore, such that $c$ is the maximum index for such an SCC
in the topologically sorted list $S_1, \ldots, S_k$,
i.e., such that
for all $d > c$, and for all $j \in S_d$, we have $v_j = 0$.

First, let us note that it must be the case that $S_c$ is a 
{\em non-trivial} SCC.
Specifically,
let us call an SCC, $S_r$ of $H$  {\em trivial} if $S_r=\{i\}$
consists of only a single variable index, $i$, and 
furthermore, such that ${\mathbf 0}=(A)_i = (P'(y))_i$,
i.e., that row $i$ of the matrix $A$ is all zero.
This can not be the case for $S_c$, because for any variable $i \in S_c$,
we have $v_i > 0$, and thus $(Av)_i = \rho(A) v_i > 0$.

Let us consider the principal submatrix $A[S_c]$ of $A$.
We claim that $\rho(A[S_c]) = \rho(A)$.
To see why this is the case, note that  $Av= \rho(A)v$,
and for every $i \in S_c$,   we have $(Av)_i = \sum_j a_{i,j} v_j = \rho(A)v_i$.
But $v_j = 0$ for every $j \in S_d$ such that $d > c$,  and
furthermore $a_{i,j} = 0$ for every $j \in S_{d'}$ such that $d' < c$.

Thus, if we let $v_{S_c}$ denote the subvector of $v$ corresponding
to the indices in $S_c$,  then we have just established that 
$A[S_c] v_{S_c} = \rho(A) v_{S_c}$, and thus that 
$\rho(A[S_c]) \geq \rho(A)$.  But since $A[S_c]$ is
a principal submatrix of $A$, we also know easily
(see, e.g, Corollary 8.1.20 of \cite{HornJohnson85}), that 
$\rho(A[S_c]) \leq  \rho(A)$, so $\rho(A[S_c]) = \rho(A)$.

We are almost done.   Given the original PPS, $x=P(x)$,
for any subset $S \subseteq \{1,\ldots,n\}$ of variable indices, 
let
$x_S = P_S(x_S,x_{D_S})$ 
denote  the subsystem of $x=P(x)$ associated with the vector $x_S$ of 
variables in set $S$, where $x_{D_S}$ denotes the variables not in $S$.

Now, note that  $x_{S_c} = P_{S_c}(x_{S_c},y_{D_{S_c}})$ is itself a PPS.
Furthermore, it is a {\em strongly connected} PPS, precisely
because $S_c$ is a strongly connected component of the 
dependency graph $H$, and because $y > 0$.
Moreover, the Jacobian matrix  of
$P_{S_c}(x_{S_c},y_{D_{S_c}}))$, evaluated at $y_{S_c}$, which we
denote by $P'_{S_c}(y)$, is precisely the principal submatrix $A[S_c]$
of $A$.  Since $x_{S_c} = P_{S_c}(x_{S_c},y_{D_{S_c}})$ is a strongly
connected PPS, we have already argued that it must be the case that
$\rho(P'_{S_c}(y)) < 1$.  Thus  since $P'_{S_c}(y) = A[S_c]$, we
have $\rho(A[S_c]) = \rho(A) < 1$.
This completes the proof.
\end{proof}

\section{Omitted Material from Section 3}

\subsection{Proof of Theorem \ref{mptime}}

\noindent {\bf Theorem  \ref{mptime}}
{\em Given any max/minPPS, $x=P(x)$, with
LFP $0 < q^* < 1$.  If we use the
``rounded-down-GNM'' algorithm with rounding parameter 
$h = j +2+4|P|$, then 
the iterations are all defined,
and for every $k \geq 0$ we have $0 \leq x^{(k)} \leq q^*$, and 
furthermore after $h = j + 2 + 4|P|$ iterations we
have:
$$\|q^*- x^{(j+2+4|P|)}\|_\infty \leq 2^{-j}$$}\\

We prove this using a few lemmas.

\begin{lemma}\label{round-exists-GNM}
If we run the rounded-down-GNM
starting with $x^{(0)} := \textbf{0}$
on a max/minPPS, $x=P(x)$, with LFP $q^*$,  ${\mathbf 0} < q^* < {\mathbf 1}$, 
then for all $k \geq 0$, $x^{(k)}$ is well-defined and $0 \leq x^{(k)} 
\leq q^*$.
\end{lemma}

\begin{proof}
The base case $x^{(0)} = 0$ is immediate for both. 

For the induction step,  suppose the claim
holds for $k$ and thus $0 \leq x^{(k)} \leq q^*$. 
From Proposition \ref{subgoal},
$I(x^{(k)})$ is well-defined and $I(x^{(k)}) \leq q^*$.
Furthermore, since $x^{(k+1)}$ is obtained from $I(x^{(k)})$ by
rounding down all coordinates, except setting to $0$ any that are negative,
and since obviously $q^* > 0$,
we have that    $0 \leq x^{(k+1)} \leq q^*$.

\end{proof}

\begin{lemma}\label{lem:explicit-bound-rounded-GNM}
For a max/minPPS, $x=P(x)$, with LFP $q^*$, such that $0 < q^* < 1$, if we
apply rounded-down-GNM with parameter $h$,
starting at $x^{(0)} := \textbf{0}$, then for all $j' \geq 0$, we have:
$$ \| q^* - x^{(j'+1)} \|_{\infty} \leq 2^{-j'} + 2^{-h+1 + 4|P|}$$
\end{lemma}
\begin{proof} 
Since $x^{(0)} := 0$:

\begin{equation}
\label{base-equation}
q^* - x^{(0)} = q^* \leq \textbf{1}   \leq \frac{1}{(\textbf{1} - q^*)_\text{min}} (\textbf{1} - q^*)
\end{equation}
For any $k \geq 0$, if
$q^*-x^{(k)} \leq \lambda (\textbf{1} - q^*)$, then
by Proposition   \ref{subgoal}(which was proved separately
for maxPPSs and minPPSs,  in Lemmas  \ref{maxhalf} 
and   \ref{minhalf}, respectively),
we have:

\begin{equation}
\label{eq:half-in-round}
q^*-I(x^{(k)})\leq (\frac{\lambda}{2}) (\textbf{1} - q^*)
\end{equation}
Observe that after every iteration $k > 0$,  in every coordinate $i$ we have:

\begin{equation}
\label{eq:almost-bigger-in-round}
x_i^{(k)} \geq I(x^{(k-1)})_i - 2^{-h}
\end{equation}
This holds simply because we are rounding down $I(x^{(k-1)})_i$ by
at most $2^{-h}$, unless it is negative in which case 
$x^{(k)}_i =0> I(x^{(k-1)})_i$.
Combining the two inequalities (\ref{eq:half-in-round}) and
(\ref{eq:almost-bigger-in-round})
 yields the following inequality:

$$q^*-x^{(k+1)}\leq (\frac{\lambda}{2}) (\textbf{1} - q^*) + 2^{-h} \textbf{1} \leq
(\frac{\lambda}{2} + \frac{2^{-h}}{(\textbf{1} - q^*)_\text{min}}) (\textbf{1} - q^*)$$

\noindent Taking inequality (\ref{base-equation}) as the base case 
(with $\lambda =
\frac{1}{(\textbf{1} - q^*)_\text{min}}$), by 
induction on $k$, for all 
$k \geq 0$:

$$q^*-x^{(k+1)} \leq (2^{-k} + \sum_{i = 0}^k 2^{-(h+i)}) \frac{1}{(\textbf{1} - q^*)_\text{min}} (\textbf{1} - q^*)$$

\noindent But $\sum_{i = 0}^k 2^{-(h+i)} \leq 2^{-h+1}$ and 
$\frac{\|\textbf{1} - q^*\|_\infty}{(\textbf{1} - q^*)_\text{min}} 
\leq \frac{1}{ (\textbf{1} - q^*)_\text{min}} \leq 2^{4|P|}$,
by Lemma \ref{gap}.  Thus:
$$q^*-x^{(k+1)} \leq (2^{-k} + 2^{-h+1})2^{4|P|} \textbf{1}$$
Clearly, we have $q^*-x^{(k)} \geq 0$ for all $k$. Thus we have shown
that for all $k \geq 0$:
$$\|q^*-x^{(k+1)}\|_\infty \leq (2^{-k} + 2^{-h+1})2^{4|P|} = 2^{-k} + 2^{-h + 1 + 4 |P|}.$$
\end{proof}

\begin{proof}[{\bf Proof of Theorem  \ref{mptime}.}]
In
Lemma \ref{lem:explicit-bound-rounded-GNM}
let $j' := j+4|P|+1$ and $h := j + 2 + 4|P|$.  We have:
$\| q^* - x^{(j + 2 + 4|P|)} \|_{\infty} \leq 2^{-(j+1+4|P|) } +
2^{-(j+1)} \leq 2^{-(j+1)} + 2^{-(j+1)} = 2^{-j}.$
\end{proof}

\section{Omitted Material from Section 4.}

\subsection{Bounds on the norm of $(I-P'(x))^{-1}$.}

We aim to prove Theorem \ref{normbounds}, which we re-state here.
Let us first recall some definitions related to the dependency graph
of variables in a PPS.

For a PPS, $x=P(x)$ with $n$ variables, its variable
{\em dependency graph} is defined to be the digraph $H = (V,E)$, with 
vertices $V=\{x_1,\ldots,x_n\}$, such that 
$(x_i,x_j) \in E$ iff in $P_i(x) \equiv \sum_{r \in R_i} p_r x^{v(\alpha_r)}$ 
there is a coefficient $p_r > 0$ such that $v(\alpha_r)_j > 0$.
Intuitively, $(x_i, x_j) \in E$ means that $x_i$ ``depends directly''
on $x_j$.
A MPS or PPS, $x=P(x)$, is called {\bf\em strongly connected} if its 
dependency graph $H$
is strongly connected.
\\

\noindent {\bf Theorem \ref{normbounds}.} {\em 
If $x=P(x)$ is a PPS with LFP $q^* > 0$  then
\begin{itemize}
\item[{\bf (i)}] If $q^* < 1$ and $0 \leq y < 1$, then 
 $(I-P'(\frac{1}{2}(y+q^*)))^{-1}$ exists and is non-negative, and
$$\|(I-P'(\frac{1}{2}(y+q^*)))^{-1}\|_\infty \leq  2^{10|P|} \text{max } 
\{2(1-y)_{\min}^{-1}, 2^{|P|}\}$$
\item[{\bf (ii)}] If $q^* = 1$ and $x = P(x)$ is strongly connected (i.e. every variable depends on every other) and $0 \leq y <  1 = q^*$, then
 $(I -P'(y))^{-1}$ exists and is non-negative, and 
$$\|(I -P'(y))^{-1}\|_\infty \leq 2^{4|P|} \frac{1}{(1-y)_{\min}}$$
\end{itemize}}

Before proving this Theorem, we shall need
to develop some more definitions and lemmas.

\begin{definition} A {\em path} in the dependency graph $H=(V,E)$ 
of a PPS $x = P(x)$ is a sequence of variables $x_{k_1}$, ... ,$x_{k_m}$,
with $m \geq 2$,
such that $(x_{k_{i}},x_{k_{i+1}}) \in E$, for $i \in \{1,\ldots,m-1\}$.  
In other words,  for each $i \in \{1,\ldots,m-1\}$, 
$x_{k_{i+1}}$ appears (with a non-zero coefficient) in the polynomial $P_{k_{i}}(x)$.

We say 
that $x_i$ {\bf\em depends on} $x_j$ (directly or indirectly) 
if there is a path in the dependency graph starting at $x_i$ and ending at $x_j$.
\end{definition}         

\noindent We shall need to be  more quantitative about dependency: 

\begin{lemma} \label{deplbound} Given a PPS $x=P(x)$ in SNF form, and variables $x_i$,$x_j$:
\begin{itemize}
\item[(i)] If $x_i$ depends on $x_j$ then there is a positive integer $k$, with $1 \leq k \leq n$,
such that $$(P'(1)^k)_{ij} \geq 2^{-|P|}$$
\item[(ii)] If $(P'(1)^k)_{ij} > 0$ for some positive
integer $k$, with $1 \leq k \leq n$, then $x_i$ depends on $x_j$.
\item[(iii)] If $x_i$ depends on $x_j$ "only via variables of Form L", i.e., 
if there is a path $x_{l_1},\ldots,x_{l_m}$ in the dependency graph such
that $l_1 = i$ and $l_m = j$, 
and such that for each $1 \leq h \leq m-1$,  
$x_{l_h} = P_{l_h}(x) = p_{l_h,0} + \sum_{g=1}^n p_{l_h,g}x_g$ has form L with 
$p_{l_h, l_{h+1}} > 0$, then there is a $1 \leq k \leq n$ such that, 
for any vector $x$, such that $0 \leq x \leq 1$, $$(P'(x)^k)_{ij} \geq 2^{-|P|}$$
\end{itemize}
\end{lemma}
\begin{proof}\mbox{}
\begin{itemize}
\item[(i)] 
Let the sequence of variables $x_{l_1}, \ldots, x_{l_k}$ 
constitute a shortest path from $x_i$ and $x_j$, such that $k \geq 2$.
Such a shortest path exists, since $x_i$ depends on $x_j$.
So $x_i = x_{l_1}$, and $x_j = x_{l_k}$, and
$x_{l_{h+1}}$ appears in the expression for $P_{l_h}(x)$,  and 
$1 \leq h \leq k-1$. 
Note that we must have $k \leq n$.
Thus $(P'(1))_{l_{h}l_{h+1}}>0$
for $1 \leq h \leq k-1$.  But note that since $P'(1)$ is a non-negative matrix,
$(P'(1)^{k-1})_{ij} \geq \prod_{h=1}^{k-1} (P'(1))_{l_{h}l_{h+1}}$. 
Since we have chosen a shortest (non-empty) path from $x_i$ to $x_j$,
and since $x=P(x)$ is in SNF form,
each $(P'(1))_{l_{h}l_{h+1}}$ that is not exactly $1$ must be a distinct rational coefficient 
in $P$, not appearing elsewhere along the path, 
and thus $\prod_{h=1}^{k-1} (P'(1))_{l_{h}l_{h+1}} \geq 2^{-|P|}$.
\item[(ii)] For $k \geq 1$, we can expand $(P'(1)^k)_{ij}$ into a sum of 
$n^{k-1}$ terms of the form 
$\prod_{h=1}^{k} (P'(1))_{l_{h}l_{h+1}}$ with $l_1=i$, $l_{k+1} = j$ and 
$(l_2,\ldots,l_{k}) \in \{1,...,n\}^{k-1}$. At least one of these has 
$\prod_{h=1}^{k} (P'(1))_{l_{h}l_{h+1}} > 0$. 
In that case, $x_{h_1},...,x_{h_{k+1}}$ is a path in the dependency graph starting at 
$x_i$ and ending at $x_j$. 
\item[(iii)]  Let us choose
$x_{l_1}, \ldots, x_{l_k}$ to be a shortest path from $x_i$ to $x_j$,
with $k \geq 2$, and such that every equation $x_{l_h} = P_{l_h}(x)$ along
the path, for all $h \in \{1,\ldots,k-1\}$ has form L.
Clearly, we must have $k \leq n$. 
By monotonicity of $P'(z)$ in $z \geq 0$, we have
$(P'(1)^{k-1})_{ij} \geq P'(x)^{k-1}$.
Furthermore, since $x_{l_1}, \ldots, x_{l_{k}}$ is a path
from $x_i$ to $x_j$, we have 
$(P'(x))^{k-1}_{i,j} \geq \prod_{h=1}^{k-1} (P'(x))_{l_{h}l_{h+1}}$. 
Moreover, since each equation 
$x_{l_h} = P(x)_{l_h}$ has Form L, for every $h \in \{1,\ldots,k-1\}$, we must have 
$(P'(x))_{l_{h}l_{h+1}} = (P'(1))_{l_{h}l_{h+1}}$  (because all the 
partial derivatives of linear expressions are constants). But we argued in (i) that,
when $x_{l_1}, \ldots, x_{l_k}$ constitutes a shortest path from $x_i$ to $x_j$,  
$\prod_{h=1}^{k-1} (P'(1))_{l_{h}l_{h+1}} \geq 2^{-|P|}$.
\end{itemize}\end{proof}

We need a basic result from the Perron-Frobenius theory of non-negative matrices.
We are not aware of a source that contains
a statement exactly equivalent to (or implying) the following Lemma, so we shall 
provide a proof, however it is entirely possible (and likely) that such a Lemma
has appeared elsewhere. 
Lemma 19 of \cite{EWY10} provides a similar result for the case when the matrix 
$A$ is irreducible.
\begin{lemma} \label{quantpf} If $A$ is a non-negative matrix, 
and vector $u > 0$ is such that $Au \leq u$ and $\|u\|_\infty \leq 1$,   
and $\alpha , \beta \in (0,1)$ are constants such that for every $i \in \{1,...n\}$, 
one of the following two conditions holds:
\begin{itemize}
\item[(I)]   $(Au)_i \leq (1 - \beta) u_i$

\item[(II)]  there is some $k$, $1 \leq k \leq n$, and  
some $j$, such that $(A^k)_{ij} \geq \alpha$ and $(Au)_j \leq (1 - \beta) u_j$.
\end{itemize} 
then $(I-A)$ is non-singular and
$$\|(I-A)^{-1}\|_\infty \leq \frac{n}{u_{\min}^2\alpha \beta}$$
\end{lemma}

\begin{proof}    First, suppose that some $i \in \{1,\ldots,n\}$,
satisfies condition $(I)$.  Then, we claim that it satisfies
condition (II), except that we must take $k=0$.
Specifically, if we let $k=0$, then since $A^0 = I$, and $(A^0)_{ii} = I_{ii} = 1 \geq \alpha$,
condition (II) boils down to 
$(Au)_i \leq (1 - \beta) u_i$.
So, to prove the statement, it suffices to only consider condition (II)
but to allow $k=0$ in that condition.

So, by assumption, given any $i \in \{1,...n\}$, there is some 
$0 \leq k \leq n$ and some $j$, such that
\begin{equation}\label{eq:first-for-the-perron-frob-lem}
(A^k)_{ij} \geq \alpha > 0
\end{equation}
and moreover $(Au)_j \leq (1 - \beta) u_j$, which we can rewrite as: 
\begin{equation}\label{eq:second-for-the-perron-frob-lem}
u_j - (Au)_j \geq \beta u_j   \ \  ( \ > 0  \ ) 
\end{equation}
Let $u_{\min} = \min_i u_i$.  We thus have that for every $i$:
\begin{eqnarray*} (A^n u)_i & = & (u - \sum_{l = 0}^{n-1} A^l  (u - Au))_i \\
			   & \leq & (u - A^k (u-Au))_i \quad \quad \quad \quad \mbox{(because $A^l \geq 0$ and $(u-Au) \geq 0$)}\\
& = &  (u_i -  \sum^n_{j'=1} A^k_{ij'}(u_{j'}-(Au)_{j'})\\
& \leq & (u_i -  A^k_{ij}(u_j-(Au)_j) \quad \quad \mbox{(again, because
$A^k_{i,j'} \geq 0$ and $(u_{j'}-(Au)_{j'}) \geq 0$ for every $j'$)}\\ 
& \leq & u_i - \alpha \beta u_j  \quad \quad \quad \quad \quad \quad \quad \mbox{(by
(\ref{eq:first-for-the-perron-frob-lem}) and (\ref{eq:second-for-the-perron-frob-lem}))}\\
	& \leq & u_i - \alpha \beta u_{\min} \\
	& \leq & u_i - u_{\min} \alpha \beta u_i \quad \quad \quad \quad \quad \quad
\mbox{(recalling that by assumption $\|u\|_{\infty} \leq 1$)}\end{eqnarray*}
We have that $A^n u \leq (1- u_{\min} \alpha \beta) u$. 
Of course $(1- u_{\min} \alpha \beta) < 1$. 
So we have that
$$A^{mn} u \leq (1- u_{\min} \alpha \beta)^{m} u$$
For any integer $d \geq 0$, $A^d u \leq u$. 
Thus also, for every $d \geq 0$,   
\begin{equation}\label{eq:simple-ineq-per-frob-prf} 
A^d u \leq 
 (1- u_{\min} \alpha \beta)^{\lfloor \frac{d}{n} \rfloor} u
\end{equation}
We thus have that, as $m \rightarrow \infty$,  $A^m u \rightarrow 0$.
Since $u > 0$ and $A \geq 0$, this implies that  as $m \rightarrow \infty$, 
$A^m \rightarrow 0$ (coordinate-wise), or in other words that 
$\lim_{m \rightarrow \infty} \| A^m \|_\infty  = 0$.
This is equivalent to saying that the spectral radius $\rho(A) < 1$.
Let us first recall that this implies 
that the inverse matrix $(I-A)^{-1} = \sum^{\infty}_{k=0} A^k \geq 0$ exists.

\begin{lemma}{(see, e.g., \cite{HornJohnson85}, Theorem 5.6.9 and Corollary 5.6.16)}
\label{series} If A is a square matrix with $\rho(A) < 1$ then $(I-A)$ is non-singular, 
the series $\sum_{k=0}^\infty A^k$ converges, and
	$(I - A)^{-1} = \sum_{k=0}^\infty A^k$.
\end{lemma}

\noindent We will use the following easy fact:

\begin{lemma} \label{pfnorm} If $M$ is a nonnegative $n \times n$ matrix, $u > 0$ is 
a vector with $\|u\|_\infty \leq 1$, and $\lambda > 0$ is a real number satisfying 
$M u \leq \lambda u$
then 
$$\|M\|_\infty \leq \frac{\lambda}{u_{\min}}$$\end{lemma}
\begin{proof} Since $M$ is non-negative,
$\|M\|_\infty$ is the maximum row sum of $M$. 
There is thus an $i$ such that
$$\|M\|_\infty = \sum_j m_{ij}$$
where $m_{i,j}$ are the entries of $M$. For this $i$:
\begin{eqnarray*} \lambda u_i & \geq & (M u)_i \\
						& = & \sum_j m_{ij} u_j\\
						& \geq & \sum_j m_{ij} u_{\min}\\
						& = & \|M\|_\infty u_{\min}\end{eqnarray*}
but $u_i \leq 1$ giving us
$||M||_\infty \leq \frac{\lambda}{u_{\min}}$.\end{proof}

\noindent Now we can complete the proof of Lemma \ref{quantpf}:

\begin{eqnarray*} (I-A)^{-1} u = (\sum_{k=0}^\infty A^k) u & = & \sum_{k =0}^\infty A^k u \\
					& \leq & \sum_{k =0}^\infty (1- u_{\min}  \alpha \beta)^{\lfloor \frac{k}{n} \rfloor} u  \quad \quad \quad 
\mbox{(by  (\ref{eq:simple-ineq-per-frob-prf}))}\\
		& = & (\sum_{m =0}^\infty n (1- u_{\min} \alpha \beta)^m  u \\
		& = & n \frac{1}{u_{\min} \alpha \beta} u \end{eqnarray*}
the last equality holding because  the geometric series sum gives
$\sum^\infty _{m=0} (1- u_{\min} \alpha \beta)^m = \frac{1}{u_{\min} \alpha \beta}$.
Lemma \ref{pfnorm}, with $M := (I-A)^{-1} = \sum^{\infty}_{k=0} A^k$,  and $\lambda :=
 n \frac{1}{u_{\min} \alpha \beta}$, now yields:
$$\|(I-A)^{-1}\|_\infty \leq n \frac{1}{u_{\min}^2\alpha \beta} $$
and this completes the proof of Lemma \ref{quantpf}.
\end{proof}

\begin{proof}[{\bf Proof of Theorem \ref{normbounds}.}]

Before we start to prove cases {\bf (i)} and {\bf (ii)} of the Theorem  we need to
develop some more lemmas.

\begin{proposition} \label{dependpos} 
For a PPS, $x=P(x)$, with LFP $q^* > 0$, for every variable $x_i$ either $P_i(0) > 0 $ or 
$x_i$ depends on a variable $x_j$ with $P_j(0) > 0$.\end{proposition}
\begin{proof} Suppose, for contradiction, that a variable $x_i$ has $P_i(0) = 0$ and depends only 
on variables $x_j$ which have $P_j(0) = 0$. Then $P^n_i(0) = 0$ for all $n$. But $P^n(0) \rightarrow q^*$ as $n \rightarrow \infty$ (see e.g.,. Theorem 3.1 from \cite{rmc}). So $q^*_i=0$.\end{proof}
\noindent The case when all the equations, $x_i = P_i(x)$, are linear has to be treated a little differently, 
and we tackle that first:

\begin{lemma} \label{Lbound} If $x=P(x)$ is a PPS that has no equations of form Q, and 
has LFP $q^* > 0$, then
$$\|(I - P')^{-1}\|_\infty \leq n2^{2|P|}$$
where $P'$ is the constant Jacobian matrix of $P(x)$,  (i.e., $P' = P'(x)$ for all $x$).\end{lemma}
\begin{proof} First, note that $P'$ is a sub-stochastic matrix i.e. $P'1 \leq 1$. 
We will now call a variable, $x_i$, {\em leaky},
if $(P' 1)_i < 1$.  Note that since $P_i(x) \equiv \sum^n_{i=1}p_{i,j}x_j + p_{i,0}$,
this means that $(P' 1)_i = \sum^n_{j=1} \frac{\partial P_i(x)}{\partial x_j} =  \sum^n_{j=1} p_{i,j} < 1$. 

Note that since $q^* > 0$, it must be the case that for every variable $x_i$, either $x_i$
itself is leaky, or $x_i$  depends (possibly indirectly) on a leaky variable $x_j$.
This is because if a variable $x_i$ doesn't satisfy this, then $q^*_i = 0$, which can't be the case.

Since the entries of $P'$ are either $0$, $1$, or coefficients $p_{i,j}$ from $P(x)$, 
we see that for every {\em leaky} variable $x_i$, we have that 
$(P'1)_i = \sum^n_{j=1} p_{i,j} \leq (1 - 2^{-|P|})$ holds.\footnote{\label{footnote:almost-end}This inequality holds because
we assume each positive input probability $p_{i,j}$ is represented as a ratio $\frac{a_j}{b_j}$
of positive integers in the encoding of $x=P(x)$, 
and thus $1- \sum^n_{j=1} \frac{a_j}{b_j}$  can be represented as a ratio $\frac{a}{b}$ 
of two positive integers where the denominator is $b = \prod^n_{j=1} b_j$.
But then  $(1- \sum^n_{j=1} \frac{a_j}{b_j}) = \frac{a}{b} \geq 1/\prod^n_{j=1} b_j \geq \frac{1}{2^{|P|}}$.}

For any {\em non-leaky} variable $x_{r}$, there is a leaky variable $x_i$ that $x_r$ depends on. 
$x_r$ does not depend on any variables of form Q. Thus, by Lemma \ref{deplbound} (iii), there is a $k$,
$1 \leq k \leq n$, 
such that $((P')^k)_{ri} \geq 2^{-|P|}$.

We can thus apply Lemma \ref{quantpf} with matrix $A := P'$ 
and vector $u:=1$,  with $\alpha:= \beta := 2^{-|P|}$,
because we have just established that 
condition (I) of that Lemma applies to leaky variables $x_i$, and
condition (II) of that Lemma applies to non-leaky variables.
Thus Lemma \ref{quantpf} 
give us that $$\|(I -P')^{-1}\|_\infty \leq (\frac{1}{1_{\min}})^2 n2^{2|P|}$$
Of course, $1_{\min} = 1$.\end{proof}

We are now ready to prove parts {\bf (i)} and {\bf (ii)} of Theorem \ref{normbounds}.

\noindent {\bf (i)} When $q^* < 1$, we can say something stronger than Proposition \ref{dependpos}.
\begin{lemma}\label{lem:almost-at-the-end}
 For any PPS, x=P(x), with LFP $0 < q^* < 1$, for
any variable $x_i$ either
\begin{itemize}
\item[(I)]  the equation $x_i = P_i(x)$ is of form Q, or else $P_i(1) < 1$.

\item[(II)] $x_i$  depends on a variable $x_j$, such that $x_j=P_j(x)$ is of form Q,
or else $P_j(1) < 1$.
\end{itemize}\end{lemma}

\begin{proof} Suppose, for contradiction, that there is a variable
  $x_i$ for which neither (I) nor (II) holds.  Let $D_i$ be the set of
  variables that $x_i$ depends on, unioned together with $\{x_i\}$ itself.  For
  any vector $x$, consider the subvector $x_{D_j}$, which consists of
  the components of $x$ with coordinates in $D_i$.  We can consider
  the subset of the equations $x_{D_i} = P_{D_i}(x)$. By transitivity
  of dependency, $P_{D_i}(x)$ contains only terms in the variables
  $x_{D_i}$. So $x_{D_i} = P_{D_i}(x) = P_{D_i}(x_{D_i})$ is itself a PPS. Since
  by assumption neither (I) nor (II) hold
  for  $x_i$, we have that $x_{D_i} = P_{D_i}(x_{D_i})$ contains no
  equations of form Q and $P_{D_i}(1) = 1$.   Since, therefore, $P_{D_i}(x_{D_i})$ is
  linear, we can rewrite $x_{D_i} = P_{D_i}(x_{D_i})$ as $x_{D_i} =
  P'_{D_i} x_{D_i} + P_{D_i}(0)$ and hence $(I-P'_{D_i}) x_{D_i} =
  P_{D_i}(0)$. Lemma \ref{Lbound} applied to the PPS $x_{D_i} =
  P_{D_i}(x_{D_i})$ gives us that, in particular, $(I-P'_{D_i})$ is
  non-singular. Consequently $x_{D_i} = P_{D_i}(x_{D_i})$ has a unique
  solution. But we already said that $1$ is a solution, $P_{D_i}(1) = 1$, and so $q^*_{D_i}
  = 1$. This contradicts $q^* < 1$. So there can be no $x_i$ 
for which neither (I) nor (II) holds.\end{proof}

To obtain the conclusion of case {\bf (i)} of Theorem \ref{normbounds}, 
assuming all of the premises of the Theorem's statement, we 
will now aim to use Lemma \ref{quantpf}, applied to $A := P'(\frac{1}{2}(y+q^*)$,
and $u := 1-q^*$. 

By Lemma \ref{lem:almost-at-the-end},
every variable $x_i$ either depends on a variable, or is itself equal to a variable, $x_j$, 
such that $x_j=P_j(x)$ is of form Q or $P_j(1) < 1$. We can 
clearly assume that such a dependence is linear in the sense of 
Lemma \ref{deplbound} (iii), and thus for any $x_i$ there is a $0 \leq k \leq n$ with 
$(P'(1)^k)_{ij} \geq 2^{-|P|}$, for some $x_j$ with either $x_j=P_j(x)$ of form Q or $P_j(1) < 1$.

We need to show and that for such an $x_j$ we have $(P'(\frac{1}{2}(y+q^*))(1-q^*) < 1-q^*$.

For any variable $x_j$ such that $x_j=P_j(x)$ has form Q, we have that $x_j = x_kx_l$ for some variables $k$ and $l$. 
Thus, since $\frac{\partial P_j(x)}{\partial x_k} = x_l$ and $\frac{\partial P_j(x)}{\partial x_l} = x_k$,
we have that:
\begin{eqnarray*} 
(P'(\frac{1}{2}(q^*+y))(1-q^*))_j	
& = & \frac{1}{2}(q^*_k + y_k)(1-q^*_l) + \frac{1}{2}(q^*_l + y_l)(1-q^*_k) \\
& = & \frac{1}{2}((q^*_k + 1) - (1- y_k))(1-q^*_l) + \frac{1}{2}((q^*_l  + 1 ) - (1- y_l))(1-q^*_k) \\
& = & \frac{1}{2} ( (q^*_k + 1)(1-q^*_l) - (1- y_k) (1-q^*_l) + (q^*_l + 1 )(1-q^*_k) - (1- y_l)(1-q^*_k) )\\
				& = & \frac{1}{2}(2 - 2q^*_kq^*_l - (1-y_l)(1-q^*_k) - (1-y_k)(1-q^*_l)) \\
			& \leq & \frac{1}{2}(2 - 2q^*_kq^*_l - (1-y)_{\min}((1-q^*)_k +(1-q^*)_l)) \\
			& \leq & \frac{1}{2}(2 - 2q^*_kq^*_l - (1-y)_{\min}((1-q^*)_k +(1-q^*)_l - (1-q^*)_k(1-q^*)_l)) \\
& = &  (1-q^*_j) -  \frac{1}{2} (1-y)_{\min}  (1-q^*_j)\\
			& = & (1-\frac{1}{2}(1-y)_{\min})(1 - q^*)_j \end{eqnarray*}
If, on the other hand, $x_j$ has $P_j(1) < 1$, then $x_j = P_j(1)$ has form L, and,
as in the proof of Lemma \ref{Lbound}, and specifically footnote (\ref{footnote:almost-end}), we 
must have 
\begin{equation}\label{eq:for-9999999}
P_j(1) \leq 1 - 2^{-|P|}
\end{equation}
We thus have that: 
\begin{eqnarray*} (P'(\frac{1}{2}(q^*+y))(1-q^*))_j	& = & \sum_{l=1}^n p_{j,l} (1-q^*)_l \\
                                                        & = & (\sum^n_{l=1} p_{j,l}) + p_{j,0} - 
(\sum^n_{l=1} p_{j,l} q^*_l) - p_{j,0}\\
                                                        & = & P_j(1) - P_j(q^*)\\
							& = & P_j(1) - q^*_j\\
                                                        & \leq & (1-2^{-|P|}) - q^*_j \quad 
\quad \quad \mbox{(by (\ref{eq:for-9999999}))}\\
							& = & (1 - q^*)_j - 2^{-|P|} \\                        
							& \leq & (1 - 2^{-|P|}) (1 - q^*)_j \end{eqnarray*}

To be able to apply
Lemma \ref{quantpf}, it only remains to show that  $P'(\frac{1}{2}(y +q^*))) (1-q^*) \leq (1-q^*)$.
But Lemma 3.5 of \cite{ESY12} established that $P'(\frac{1}{2}(1 +q^*))) (1-q^*) \leq (1-q^*)$.
Since $0 \leq y < 1$,  it follows by monotonicity of $P'(z)$ in $z$ that  
$P'(\frac{1}{2}(y +q^*))) (1-q^*) \leq (1-q^*)$.

Thus, we can apply Lemma \ref{quantpf}, by setting $A:= P'(\frac{1}{2}(y +q^*))$,  
$u:= (1-q^*)$,   $\alpha:= 2^{-|P|}$,  $\beta := \min \{ \frac{1}{2}(1-y)_{\min},  2^{-|P|} \}$,
and we obtain:
$$\|(I -P'(\frac{1}{2}(y +q^*)))^{-1}\|_\infty \leq n (1-q^*)_{\min}^{-2} \text{max } \{2(1-y)_{\min}^{-1}, 2^{|P|}\} 2^{|P|}$$
Recall that, by Lemma \ref{1-qbound}, $(1-q^*)_{\min} \geq 2^{-4|P|}$.  Thus
\begin{eqnarray*}
\|(I -P'(\frac{1}{2}(y +q^*)))^{-1}\|_\infty & \leq &  n2^{9|P|} \text{max } \{2(1-y)_{\min}^{-1}, 2^{|P|}\}\\
& \leq & 2^{10 |P|}\text{max } \{2(1-y)_{\min}^{-1}, 2^{|P|}\}
\end{eqnarray*}

\noindent We now prove part {\bf (ii)} of Theorem \ref{normbounds}.
 If $x = P(x)$ is strongly connected, then if there is an $x_i$ with $x_i = P_i(x)$ of form Q, then every variable depends on it. If there are no such variables, then Lemma \ref{Lbound} gives that, 
for any $x \in \mathbb{R}^n$, $\|I - P'(x)\|_\infty \leq n2^{2|P|}$ and we are done. So we can assume that there is an $x_i$ with $x_i = P_i(x)$ of form Q. We quote the following from \cite{rmc}:

\begin{lemma}[see proof of Theorem 8.1 in \cite{rmc}] If $x = P(x)$ is strongly connected and 
$q^* > 0$, then $q^* = 1$ iff $\rho(P'(1)) \leq 1$.\end{lemma}
$P'(1)$ is a non-negative irreducible matrix. Perron-Frobenius theory gives us that there is a 
positive eigenvector $v > 0$, with associated eigenvalue $\rho(P'(1))$, the spectral radius of $P'(1)$, 
i.e., such that $P'(1)v = \rho(P'(1))v$. But $\rho(P'(1)) \leq 1$  so $P'(1)v \leq v$.

\begin{lemma}[cf Lemma 5.9 of \cite{lfppoly}] $\frac{\|v\|_\infty}{v_{\min}} \leq 2^{|P|}$.
\end{lemma}             
\begin{proof} For any $x_i$, $x_j$, there is some $1 \leq k \leq n$ with $(P'(1)^k)_{ij} > 0$. We know that $P'(1)^k v \leq v$. So $(P'(1)^k)_{ij} v_j \leq (P'(1)^k v)_i = \rho(P'(1))^k v_i \leq v_i$. But by Lemma \ref{deplbound} (ii), 
$(P'(1)^k)_{ij} \geq 2^{-|P|}$. So $\frac{v_j}{v_i} \leq 2^{|P|}$. There are $v_i$,$v_j$ that achieve $v_i = v_{\min}$ and $v_j = \|v\|_\infty$, so we are done. \end{proof}
We can normalise the top eigenvector, $v$, so we can assume that $\|v\|_\infty = 1$. 
Then $v_{\min} \geq 2^{-|P|}$. 
Consider any equation $x_i = P_i(x)=x_jx_k$ of form Q (we have already dealt with the case 
where no such equation exists):
\begin{eqnarray*} (P'(y)v)_i & = & y_j v_k + y_k v_j \\
& \leq &  y_{\max} v_k + y_{\max} v_j  \quad \quad  \mbox{(where $y_{\max} := \max_r y_r$)}\\
			   & \leq & (1-(1-y)_{\min}) (v_k + v_j) \\
& = & (1-(1-y)_{\min}) (P'(1)v)_i\\
				& = & (1-(1-y)_{\min}) \rho(P'(1)) v_i \\
				& \leq & (1- (1-y)_{\min})v_i  \quad \quad  \mbox{(because  $\rho(P'(1)) \leq 1$)}
\end{eqnarray*}
Now we can apply
Lemma \ref{quantpf}, with $A := P'(y)$,  $u := v$,  $\alpha := 2^{-|P|}$, and 
$\beta :=  (1-y)_{\min}$,   to obtain that:
$$\|(I-P'(y))^{-1}\|_\infty \leq n v_{\min}^{-2} (1-y)_{\min}^{-1} 2^{|P|}$$
Inserting our bound for $v_{\min}$, namely $v_{\min} \geq 2^{-|P|}$, yields:
\begin{eqnarray*}
\|(I-P'(y))^{-1}\|_\infty & \leq & n 2^{3|P|} (1-y)_{\min}^{-1}\\
& \leq &  2^{4|P|} (1-y)_{\min}^{-1}
\end{eqnarray*}
\end{proof}

\bibliographystyle{plain}
\end{document}